\documentclass{sig-alternate}

\usepackage{etex}
\usepackage{fixltx2e}

\usepackage[utf8]{inputenc}
\usepackage[T1]{fontenc}
\DeclareMathAlphabet{\mathcal}{OMS}{cmsy}{m}{n}
\usepackage{microtype}
\usepackage{makecell}
\usepackage{xcolor}
\usepackage{paralist}
\usepackage[noend]{algpseudocode}
\usepackage{algorithm}
\usepackage{amsmath}%
\usepackage{amssymb}%
\usepackage{graphicx}%
\usepackage{booktabs}
\usepackage{tikz}
\usepackage{pgfkeys}
\usepackage{pgfplots}
\usepackage{url}
\usepackage{caption}
\usepackage{subcaption}
\usepackage{multirow}
\usepackage{mathtools}
\usepackage{verbatim}
\usepackage{balance}
\usepackage[sort,compress]{cite}

\usepackage[pdfborder={0 0 0},pdfpagelabels,plainpages=false]{hyperref}

\newtheorem{definition}{Definition}

\newtheorem{lemma}{Lemma}

\newtheorem{problem}{Problem}

\numberwithin{equation}{section}

\usetikzlibrary{%
  arrows,%
  calc,%
  fit,%
  patterns,%
  plotmarks,%
  shapes.geometric,%
  shapes.misc,%
  shapes.arrows,%
  shapes.callouts,%
  shapes.multipart,%
  shapes.gates.logic.US,%
  shapes.gates.logic.IEC,%
  er,%
  automata,%
  backgrounds,%
  chains,%
  topaths,%
  trees,%
  petri,%
  mindmap,%
  matrix,%
  calendar,%
  folding,%
  fadings,%
  through,%
  positioning,%
  scopes,%
  decorations.fractals,%
  decorations.shapes,%
  decorations.text,%
  decorations.pathmorphing,%
  decorations.pathreplacing,%
  decorations.footprints,%
  decorations.markings,%
  shadows}

\usepgfplotslibrary{colorbrewer}

\pgfplotsset{every axis/.append style={font=\scriptsize, line width=0.5pt, tick
  style={line width=0.2pt}}}
\pgfplotsset{compat=1.3}
\tikzset{every mark/.append style={scale=1}}

\tikzset{line/.style=}

\DeclareMathOperator*{\argmax}{arg\,max}

\renewcommand{\epsilon}{\varepsilon}
\renewcommand{\phi}{\varphi}
\renewcommand{\leq}{\leqslant}
\renewcommand{\geq}{\geqslant}

\makeatletter
\let\@copyrightspace\relax
\makeatother

\begin{document}


\title{Online Influence Maximization}\subtitle{Extended Version}

\numberofauthors{5}
\author{
\alignauthor
  Siyu Lei\\\
  \affaddr{University of Hong Kong\\
  Pokfulam Road, Hong Kong}\\
  \email{sylei@cs.hku.hk}
\alignauthor
  Silviu Maniu\thanks{Work mainly done while the author was affiliated with
  University of Hong Kong.}\\
  \affaddr{Noah's Ark Lab, Huawei\\
  Science Park, Hong Kong}\\
  \email{silviu.maniu@huawei.com}
\alignauthor
  Luyi Mo\\
  \affaddr{University of Hong Kong\\
  Pokfulam Road, Hong Kong}\\
  \email{lymo@cs.hku.hk}
\and
\alignauthor
  Reynold Cheng\\
  \affaddr{University of Hong Kong\\
  Pokfulam Road, Hong Kong}\\
  \email{ckcheng@cs.hku.hk}
\alignauthor
  Pierre Senellart\\
  \affaddr{\makebox[0pt]{T\'el\'ecom ParisTech; CNRS LTCI}\\
  \& NUS; CNRS IPAL}\\
  \email{pierre@senellart.com}
}

\maketitle


\begin{abstract}


Social networks are commonly used for marketing purposes. 
For example, free samples of a product can be given to a few influential social network users (or ``seed nodes''), 
with the hope that they will convince their friends to buy it. One way to formalize marketers' objective is through \emph{influence
maximization} (or IM), whose goal is to find the best seed nodes to activate under a fixed budget, 
so that the number of people who get influenced in the end is maximized.
Recent solutions to IM rely on the {\it influence probability} that a
user influences another one. However, 
this probability information may be unavailable or incomplete.


In this paper, we study IM in the absence of complete information on influence probability. We call this problem 
\emph{Online Influence Maximization} (OIM) since we learn influence probabilities
at the same time we run influence campaigns. To solve OIM, we propose a
multiple-trial approach, where (1) some seed nodes are selected based
on existing influence information; (2) an influence campaign is started
with these seed nodes; and (3) users' feedback is used to
update influence information. We adopt the Explore--Exploit strategy,
which can select seed nodes using either the current influence probability estimation (exploit), or the confidence bound on the estimation (explore).
Any existing IM algorithm can be used in this framework. We also develop
an incremental algorithm that can significantly reduce the overhead of
handling users' feedback information. Our experiments show that our
solution is more effective than traditional IM methods on the
partial information.


\end{abstract}


\section{Introduction}
\label{sec:intro}

In recent years, there has been a lot of interest about how social network users
can affect or {\it influence} others (via the so-called \emph{word-of-mouth}
effect). This phenomenon has been found to be useful for marketing purposes. For
example, many companies have advertised their products or brands on social
networks by launching influence campaigns, giving free products to a few
influential individuals (seed nodes), with the hope that they can promote the
products to their friends~\cite{marketing}. The objective is to identify a set
of most influential people, in order to attain the best marketing effect. This
problem of \emph{influence maximization} (IM) has attracted a lot of research
interest~\cite{domingos01, domingos02, wei09, wei10, wei10lt}.

Given a {\it promotion budget}, the goal of IM is to select the best seed nodes
from an {\it influence graph}.  An influence graph is essentially a graph with
{\it influence probabilities} among nodes representing social network users. In
the \emph{independent cascade model}, for example, a graph edge $e$ from user
$a$ to $b$ with influence probability $p$ implies that $a$ has a chance $p$ to
affect the behavior of $b$ (e.g., $a$ convinces~$b$ to buy a movie
ticket)~\cite{Kempe:2003}. Given an influence graph, IM aims to find $k$ {\it
seed nodes}, whose expected number of influenced nodes, or \emph{influence
spread}, is maximized. Marketing efforts can then be focused on the $k$ nodes
(or persons). In the IM literature, these seed nodes are said to be
\emph{activated}~\cite{domingos01, domingos02, wei09, wei10, wei10lt}.

\begin{figure}[t!]
  \centering
  \includegraphics[width=\linewidth]{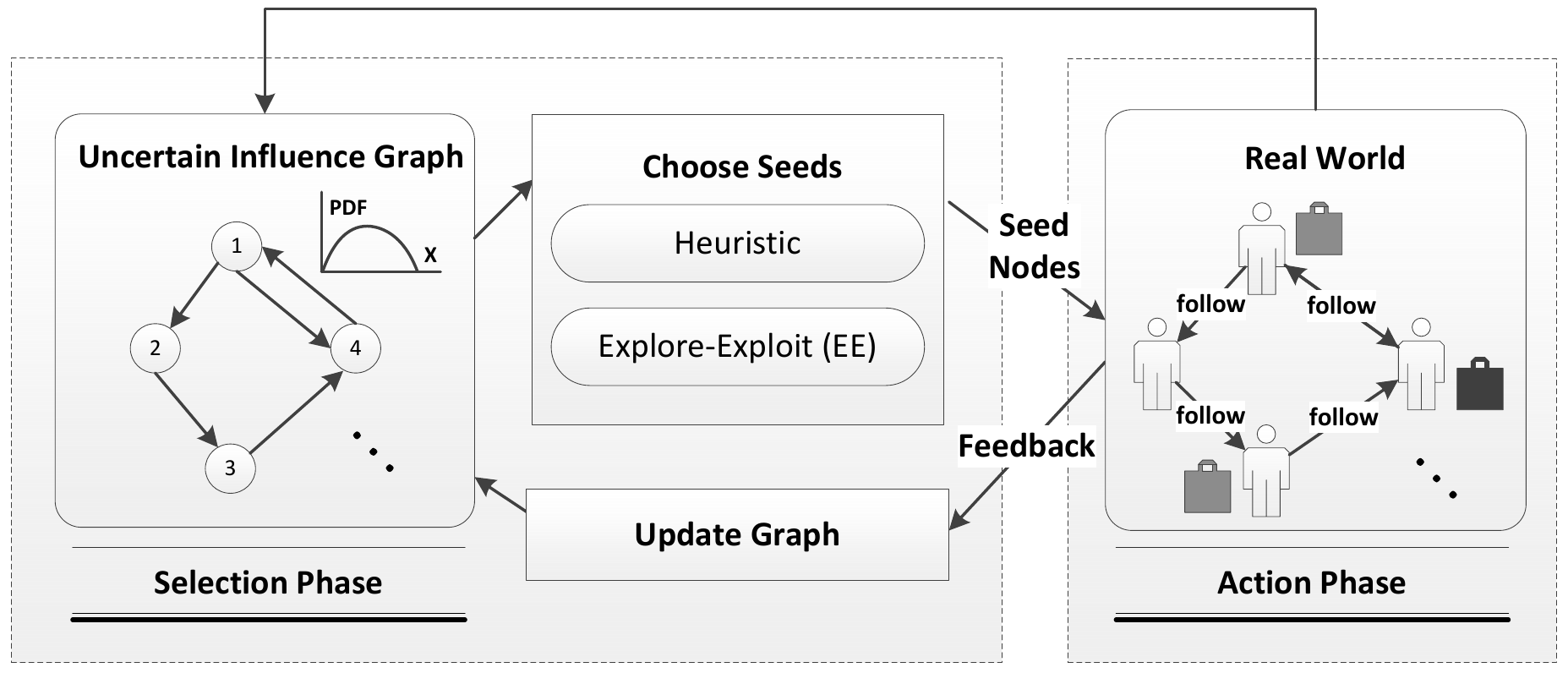}\\
  \caption{The OIM framework.}\label{sf}
  \vspace{-15.5pt}
\end{figure}

While existing IM algorithms effectively obtain the most influential seed nodes,
they assume that the influence probability value between each pair of nodes is
known. However, this assumption may not hold. Consider a marketing firm starting
in a new city with some knowledge of the social network of the users in the
city.  The company, however, does not know how influence propagates among these
users.  Unless the influence probability information is known, the marketing
firm cannot run an IM algorithm and decide the target users. To obtain these
values, {\it action logs}, which record the social network user's past
activities, can be used~\cite{Goyal:2010}. This information may not be readily
available.

Is it possible to perform IM on a social network, even if the information about
influence probabilities is {\it absent} or {\it incomplete}?  We call this
problem {\it Online Influence Maximization} (OIM), as we aim at discovering
influence probabilities at the same time we are performing influence campaigns.
(We say that an IM algorithm is \emph{offline}, if it assumes that the influence
probability between every node pair is known in advance.) In the absence of
complete influence probability information, making the best marketing effort out
of a limited promotion budget can be challenging. To tackle this problem, we
propose a solution based on influencing seed nodes in {\it multiple} rounds. In
each round, we select some seed nodes to activate (e.g., advertising a product
to a few selected users). The feedback of these users is then used to decide the
seed nodes to be activated in the next round. The information about influence
probabilities in the social network is learnt and refined during these
campaigns.

Figure~\ref{sf} illustrates our OIM framework. It contains multiple successive
influence campaigns, or {\it trials}. A trial should fulfill one of two
objectives: (1) to advertise to promising nodes; and (2) to improve the
knowledge about influence probabilities. A trial consists of two phases: {\it
selection} and {\it action}. In the {\it selection phase}, an {\it uncertain
influence graph} is maintained. This graph models the uncertainty of influence
probabilities among social network users, in terms of a probability
distribution. A {\it seed selection strategy}, based on an existing IM solution,
is then executed on this graph to produce up to $k$ seed nodes. In the {\it
action phase}, the selected seed nodes are activated in the \emph{real world}
(e.g., sending the advertisement message to chosen users). The actions of these
users, or {\it feedback} (e.g., whether the message is further spread), is then
used to {\it update} the uncertain influence graph. The iteration goes on, until
the marketing budget is exhausted. In this paper, we focus on the two crucial
components of the selection phase: (1) {\it seed selection strategy}; and (2)
techniques for {\it updating} the uncertain influence graph.

\noindent {\bf 1. Seed selection strategy.} To choose seed nodes in a trial, a simple way is to make use of existing IM algorithms.
Due to the lack of knowledge about influence probabilities, this approach may not be the best. We thus develop an {\it Explore--Exploit strategy} (or {\it EE}), which performs IM based on existing influence probability information:

\vspace{0.05in}

\noindent $\bullet$ {\bf [Exploit]} Select $k$ seed nodes for getting the most rewarding influence spread from the influence graph, derived from the uncertain influence graph. Any state-of-the-art IM algorithms (e.g., CELF~\cite{celf}, DD~\cite{wei09}, TIM and TIM+~\cite{TIM}) can be used; or

\noindent $\bullet$ {\bf [Explore]}  Select $k$ seed nodes based on some
strategy (e.g., through estimating the confidence bound of
the influence probability) to improve the knowledge about the influence graph.

\vspace{0.05in}

In this paper, we study strategies for exploit and explore. With suitable use of strategies,  {\it EE} performs better than running an existing IM algorithm on the uncertain influence graph alone.

In our OIM solution, $N$ trials are carried out. In each trial, an
existing IM algorithm may be executed. If $N$ is large, the performance
of our algorithm can be affected. The problem is aggravated if the
underlying uncertain influence graph is big. For state-of-the-art IM
algorithms (e.g., CELF~\cite{celf} and TIM+~\cite{TIM}), this running
time is dominated by the cost of sampling the influence graph. For
example, in TIM+, the sampling effort costs more than 99\% of the
computation time. We design an efficient solution, based on the intuition
that users' feedback often only affects a small portion of the influence
graph. If samples of the previous iterations are stored, it is possible to reuse them, instead of sampling the influence graph again. We examine conditions allowing a sampled graph to be effectively reused in a new trial.  We propose an \emph{incremental algorithm}, and present related data structures for facilitating efficient evaluation of our solution. This algorithm can support any sample-based IM algorithm running on independent cascade models. We demonstrate how to use TIM+ in this paper.

\noindent {\bf 2. Updating the uncertain influence graph.} As discussed before,
our seed selection strategy is executed on the uncertain influence graph
(Figure~\ref{sf}). It is important that this graph accurately reflects the
current knowledge about the influence among different users, so that the seed
selection strategy can make the best decision. We investigate algorithms for
updating this graph based on the feedback of activated users (e.g., whether they
spread out an advertisement message). We examine two variants, which update the
influence graph {\it locally} and {\it globally}. A local update refreshes the
parameters of the influence probability distribution between two graph nodes,
while a global update is applied to the parameters of the influence probability
information that applies to the whole uncertain influence graph. These
algorithms are developed based on classical machine learning methods (e.g, Least
Squares Estimation and Maximum Likelihood).

Our solutions can be adopted by social marketers who aim to promote their
products, in cases when the underlying probabilities of the influence graph are
unknown. Our approach can utilize any state-of-the-art IM algorithm. We also
examine how to update the uncertain influence graph effectively by machine
learning methods. We develop an incremental algorithm to improve the efficiency
of our solution. Our experiments demonstrate that our proposed methods can
effectively and efficiently maximize influence spread.

\section{Related Work}\label{sec:related}

\textbf{Influence Maximization (IM).}
Kempe et al.~\cite{Kempe:2003} first proposed the study of IM in social
networks. They showed that finding the set of seed nodes that maximizes
influence is NP-hard, and showed that the greedy algorithm has a constant
approximation guarantee. However, this solution is not very fast,  because
thousands of samples are often required, and each sampling operation has a
complexity linear to the graph size.  To improve the efficiency of IM solutions,
several heuristics were developed, namely Degree Discount~\cite{wei09},
PMIA~\cite{wei10}, IPA~\cite{IPA}, and IRIE~\cite{IRIE}. Although these
heuristics are fast, their accuracy is not theoretically guaranteed. Improved
aproximation algorithms with theoretical guarantees include CELF~\cite{celf},
CELF++~\cite{celfplus}, and NewGreedy~\cite{wei09}.  More recently, Borgs et
al.\ proposed an algorithm based on reverse influence sampling, and showed that
it is runtime-optimal with accuracy guarantees~\cite{Borgs}. The scalability of
this solution was enhanced by Tang et al., who developed TIM and TIM+~\cite{TIM}
to further reduce the number of samples needed.

There are also other works that address different variants of the IM
problem: (1) incorporating community~\cite{cbtopkin} and
topic~\cite{INFLEX} information in the propagation process; (2)
competition of different parties for influence~\cite{Lu:2013}; and (3)
use of other influence propagation models such as linear threshold or
credit distribution~\cite{Goyal:2011, Singer:2012, siandposteriore}.

\textbf{Learning influence probabilities.}
Saito et al.~\cite{saito} modeled the problem of obtaining influence
probabilities as an instance of likelihood maximization, and developed an
expectation maximization algorithm to solve it.  Given a social network and an
{\it action log} (e.g., user~$u$ performs action~$a$ at time~$t$), Goyal et
al.~\cite{Goyal:2010} developed static and time-dependent models to compute
influence probabilities between a pair of social network users.  These methods
require the action log information of all the users involved to be known in
advance; however, this information may not be available.  Our framework does not
require all action logs to be available. Instead, we select seed nodes in
multiple advertising campaigns, so that influence maximization can be done
faster. We then use users' feedback in each campaign to learn and refine
influence probabilities.

\textbf{Multi-armed bandits (MAB).} The EE strategy in the seed selection phase of our solution  is inspired
by the $\epsilon$-greedy algorithm, which was originally developed to solve
the \emph{multi-armed bandit problem} (MAB)~\cite{Robbins:1952}.
In the $\epsilon$-greedy algorithm~\cite{rl}, $\epsilon$~controls the trade-off between exploitation and exploration.
Specifically, with probability 1 - $\epsilon$, an action is executed based on the current knowledge (i.e., \emph{exploit}); with
probability $\epsilon$, another action is performed (i.e., \emph{explore}). This
framework is adopted as a baseline in our solution.

\cite{CMAB} studies combinatorial MAB algorithms, and in particular the {\tt
CUCB} algorithm, which uses upper confidence bounds~\cite{UCB} for choosing
between explore and exploit. A scenario akin to the OIM problem is illustrated
and it is shown that {\tt CUCB} achieves a bound on the regret. However, {\tt
CUCB} is not applicable due to two factors. First, the activated nodes are
counted multiple times leading to redundant activations and choices. Second, and
most practically important, the approximation bound depends on an initialization
step in which each arm (in this scenario, seed node) is tested to get an
activation feedback; this is not practically feasible in cases when activation
budgets are limited.  Another algorithm closely related to our framework is
Thompson Sampling~\cite{Agrawal:2012}, where each independent arm is simulated
by a Beta distribution of successes and failures. In our scenario, the
arms are the parameters of the algorithms, and defining success and failure in a
result of an influence maximization is not trivial.

\section{Influence Maximization: Review}
\label{sec:setting}

We now provide a review of the IM problem and its solutions. This forms the
basis of the OIM problem to be studied in this paper.  Table~\ref{tab:notation}
shows the notation used.

\begin{table}[t]
  \small
  \caption{Symbols used in this paper.\label{tab:notation}}
  \vspace{-1em}
  \begin{tabular}{cl}
    \toprule
    {\bf symbol} & {\bf description} \\
    \midrule
    $G$ & influence graph  \\
    $V$ & set of users (nodes) of $G$ \\
    $E$ & set of edges of $G$  \\
    $p_{ij}$ & influence probability from $i$ to $j$ (fixed
  value)  \\

  $P_{ij}$ & influence probability from $i$ to $j$ (random
variable)  \\
    $N$ & number of trials  \\
    $k$ & budget for each trial \\
    $S$ & set of seed nodes  \\
    $\sigma(S)$ & expected influence spread  \\
    $(\alpha,\beta)$ & global prior for the beta distribution  \\
    $A_n$ & set of successfully activated nodes in trial $n$ \\
    $F_n$ & real world activation feedback in trial $n$ \\
$(h_{ij},m_{ij})$ & number of successful and unsuccessful\\&
activations of the edge from $i$ to $j$\\
    \bottomrule
  \end{tabular}
\end{table}

Let $G=(V,E,p)$ be an \emph{influence graph}, where $v\in V$ are \emph{users} or
nodes, and $e\in E$ are the links or edges between them.  Each edge $e=(i,j)$
between users $i$ and $j$ is associated with an \emph{influence probability}
$p_{ij}\in[0,1]$. This value represents the probability that user $j$ is
activated by user $i$ at time $t+1$, given that user $i$ is activated at time
$t$. We also suppose that time flows in discrete, equal steps. In the IM
literature, $p_{ij}$ is given for every $i$ and $j$.  Obtaining $p_{ij}$
requires the use of action logs~\cite{Goyal:2010} which may not be available. In
this paper, we investigate how to perform IM without knowing $p_{ij}$ in
advance.

In the independent cascade model, at a given timestamp $t$, every node is in
either active (influenced) or inactive state, and the state of each node can be
changed from inactive to active, but not vice-versa. When a node~$i$ becomes
active in step $t$, the influence is independently propagated at $t+1$ from node
$i$ to its currently inactive neighbors with probability $p_{ij}$.
Node $i$ is given one chance to activate its inactive neighbor. The
process terminates when no more activations are possible. A node can be
independently activated by any of its (active) incoming neighbors. Suppose that
the activation process started from a set $S$ of nodes. We call the expected
number of activated nodes of $S$ the \emph{expected influence spread},
denoted
$\sigma(S)$.
Formally:
\newcommand{\infl}{\mathop{\mathrm{infl}}}
\begin{definition}
  Given a weighted graph $G=(V,E,p)$, let $\infl$ be the {\it immediate
  influence operator}, which is the random process that extends a set of nodes
  $X\subseteq V$ into a set of immediately influenced nodes $\infl(X)$, as
  follows:
  \[\Pr(v\in\infl(X))=\begin{cases}
      1&\text{if $v\in X$;}\\
      1-\prod_{\substack{(u,v)\in E\\u\in X}} (1-p_{uv})&\text{otherwise.}
  \end{cases}\]
  Given a seed set $S\subseteq V$, we define the
  \emph{set of influenced nodes} $I(S)\subseteq V$ as the random variable
  that is the fixpoint $I^\infty(S)$ of the following inflationary random process:
  \[\left\{  \begin{aligned}
    I^0(S)&=\emptyset;\\
    I^1(S)&=S;\\
    I^{n+2}(S)&=I^{n+1}(S)\cup \infl(I^{n+1}(S) \backslash
        I^n(S))&\text{for $n\geq 0$.}\\
    \end{aligned}\right.\]
  The influence spread $\sigma(S)$ is $\mathbb{E}[|I(S)|]$.
\end{definition}

Based on the above definition, \cite{Kempe:2003} defines the \emph{influence
maximization problem} (IM) as follows.
\begin{problem}
\label{problem:offline}
  Given a weighted graph $G=(V,E,p)$ and a number $1\leq k\leq |V|$, the
  \emph{influence maximization} (IM) problem finds a set $S\subseteq V$
  such that $\sigma(S)$ is maximal subject to $|S|=k$.
\end{problem}
As discussed in~\cite{Kempe:2003}, evaluating the influence spread is difficult.
Even when the spread values are known, obtaining an exact solution for the IM
problem is computationally intractable. Next we outline the existing IM
algorithms for this problem.

{\bf IM algorithms.} A typical IM algorithm evaluates the \emph{score} of a node
based on some metric, and inserts the $k$ \emph{best} nodes, which have the
highest scores, into $S$. For example, the degree discount (DD)
heuristic~\cite{wei09} selects the nodes with highest degree as $S$.  Another
classical example is \emph{greedy}: at each step, the \emph{next best node}, or
the one that provides the largest marginal increase for $\sigma$, is inserted
into $S$. This is repeated until $|S|=k$. The \emph{greedy} algorithm provides
an $(1-1/{e})$-approximate solution for the IM problem.  To compute the
influence spread efficiently,  \emph{sampling-based} algorithms with theoretical
guarantees were developed. For example, CELF~\cite{celf} evaluates the expected
spread of nodes with the seed nodes, and select the nodes with the largest
marginal spread; TIM~\cite{TIM} counts the frequencies of the nodes appearing in
the reversed reachable sets, and chooses the nodes with the highest frequencies;
TIM+~\cite{TIM} is an extension of TIM for large influence graphs.

We say that the above IM algorithms are \emph{offline}, since they are
executed on the influence graph once, assuming knowledge of $p_{ij}$ for every
$i$ and $j$. If these values are not known, these algorithms cannot be executed.
This problem can be addressed by \emph{online} IM algorithms, as we will discuss next.

\section{Maximizing Influence Online}
\label{sec:ugraph}

The goal of the {\it online influence maximization} (or OIM) is to perform IM
without knowing influence probabilities in advance. Given a number $N$ of
advertising campaigns (or {\it trials}), and an advertising budget of $k$ units
per trial, we would like to select up to $k$ seed nodes in each trial. These
chosen nodes are then advertised or {\it activated}, and their feedback is used
to decide the seed nodes in the next trial. Let us formulate the OIM problem
below.
\begin{problem}
\label{problem:online}
  Given a weighted graph $G=(V,E,p)$ with unknown probabilities $p_{uv}$, and a
  budget consisting of $N$ trials with $1\leq k\leq |V|$ activated nodes per trial, the
  \emph{online influence maximization} (OIM) problem is to find for each $1\leq n\leq
  N$ a set $S_n$ of nodes, with $|S_n|\leq k$, such that
  $\mathbb{E}\left[\left|\bigcup_{1\leq n\leq N} I(S_n)\right|\right]$ is
  maximal.
\end{problem}

Note that the IM problem, discussed in Section~\ref{sec:setting}, is a special
case of the OIM problem (by setting $N=1$). Since solving the IM problem is
computationally difficult, finding a solution for the OIM is also challenging.
We propose a solution that consists of multiple {\it trials}. In each trial, a
{\it selection} (for choosing appropriate seed nodes) and an {\it action}  (for
activating the seed nodes chosen) is performed (Figure~\ref{sf}). The seed
selection makes use of one of the offline IM algorithms discussed in
Section~\ref{sec:setting}.\footnote{In this paper we assume that the advertising
budget $k$ is fixed for each trial.}

We next present the {\it uncertain influence graph}, which captures the
uncertainty of influence probabilities (Section~\ref{subsec:ugraph}). We then
discuss our solution based on this graph in Section~\ref{subsec:framework}.

\subsection{The Uncertain Influence Graph}
\label{subsec:ugraph}

We assume that a social network, which describes the relationships among social
network users, is given.  However, the exact influence probability on each edge
is not known. We model this by using the {\it uncertain influence graph}, in
which the influence probabilities of each edges are captured by probability
density functions, or \emph{pdf} (Figure~\ref{sf}).  The pdf can be refined
based on the feedback returned from a trial. Since influence activations are
binary random variable, we capture the uncertainty over the influence as a
\emph{Beta distribution}. Specifically, the random variable of the influence
probability from node $i$ to node $j$, $P_{ij}$ is modeled as a Beta
distribution having probability density function:
\[
  f_{P_{ij}}(x)=
\frac{x^{\alpha_{ij}-1}(1-x)^{\beta_{ij}-1}}{B(\alpha_{ij},\beta_{ij})},
\]
where $\mathrm{B}(\alpha_{ij},\beta_{ij})$ is the Beta function, acting as a
normalization constant to ensure that the total probability mass is $1$, and
$\alpha_{ij}$ and $\beta_{ij}$ are the distribution parameters. For the Beta
distribution, $\mathbb{E}[P_{ij}] = \frac{\alpha_{ij}}{\alpha_{ij}+\beta_{ij}}$
and $\sigma^2[P_{ij}] = \frac{\alpha_{ij}
\beta_{ij}}{(\alpha_{ij}+\beta_{ij})^{2}(\alpha_{ij}+\beta_{ij}+1)}$.  An
advantage of using the Beta distribution is that it is a conjugate prior for
Bernoulli distributions, or more generally, binomial distributions. This allows
us to compute the posterior distributions easily when new evidence is provided.
Section~\ref{sec:update} explains this in more detail.

At the time of the first trial, we assume no prior information
about the influence graph, except global $\alpha$ and~$\beta$ parameters, shared
by all edges, i.e., $P_{ij}\sim \mathrm{B}(\alpha, \beta)\ \forall (i,j)\in E$.  These
global $\alpha$ and $\beta$ parameters represent our global \emph{prior belief}
of the uncertain influence graph. In the absence of any better prior, we
can set $\alpha=\beta=1$, with $B(1,1)$ being the uniform distribution.

Our model can be extended to handle various prior
information about the influence graph.  For example, if we have individual prior
knowledge ($\alpha_{ij}, \beta_{ij}$) about an edge, we can set $P_{ij}$ as
$P_{ij}\sim \mathrm{B}(\alpha_{ij}, \beta_{ij})$. When we have access
to only the mean and variance of the influence of an edge, we can derive
$\alpha_{ij}$ and $\beta_{ij}$ from the formulas of $\mathbb{E}[P_{ij}]$ and
$\sigma^2[P_{ij}]$ given above.  For the situation in which some action logs
involving the social network users are available, algorithms for learning the influence
probabilities from these logs~\cite{Goyal:2010, Goyal:2011} can be first applied, and
the estimated influence probabilities can then be used as prior
knowledge for the graph.

\subsection{The OIM Framework}
\label{subsec:framework}

\begin{algorithm}
\caption{{\tt Framework(G, k, N)}}
\label{fm}
{\footnotesize
\begin{algorithmic}[1]
    \State{\bf Input:} \# trials $N$, budget $k$, uncertain influence graph $G$
    \State{\bf Output:} seed nodes $S_n (n=1\dots N)$, activation results $A$

    \State{$A\gets\emptyset$}
    \For{$n=1$ \textbf{to} $N$}
        \State $S_n \gets \mathtt{Choose}(G,k)$ \label{ln:Choose}
        \State $(A_n, F_n)\gets \mathtt{RealWorld}(S_n)$ \label{ln:RW}
        \State $A\gets A\cup A_n$ \label{ln:UpdateFeedback}
        \State $\mathtt{Update}(G,F_n)$ \label{ln:UpdateModel}
	\EndFor
\State \Return $\{S_n | {n=1\dots N}\}$, $A$
\end{algorithmic}}
\end{algorithm}

Algorithm~\ref{fm} depicts the solution framework of the OIM problem.  In this
algorithm, $N$ trials are executed. Each trial involves selecting seed nodes,
activating them, and consolidating feedback from them. In each trial $n$ (where
$n=1,\ldots,N$), the following operations are performed on the uncertain influence
graph $G$:

\begin{compactenum}
  \item {\tt Choose} (Line~\ref{ln:Choose}): A seed set $S_n$ is chosen from $G$, by using an offline IM algorithm, and strategies for
    handling the uncertainty of $G$ (Section~\ref{sec:choose}).
  \item {\tt RealWorld} (Lines~\ref{ln:RW}--\ref{ln:UpdateFeedback}): The selected seeds set is tested in the real world (e.g.,
    sending advertisement messages to selected users in the social network). The
    feedback information from these users is then obtained. This is a tuple
    $(A_n,F_n)$ comprised of: \begin{inparaenum}[(i)] \item the
      set of activated nodes $A_n$, and \item the set of edge activation
        attempts $F_n,$ which is a list of edges having either a
        successful or an unsuccessful activation.  \end{inparaenum}
   \item {\tt Update} (Line~\ref{ln:UpdateModel}): We refresh $G$ based on $(A_n,F_n)$ (Section~\ref{sec:update}).
\end{compactenum}

One could also choose not to update $G$, and instead only run an offline IM
based on the prior knowledge. Our experimental results show that the influence
spread under our OIM framework with proper updates is better than the one
without any update. Next, we investigate the design and implementation of {\tt
Choose} (Section~\ref{sec:choose}) and {\tt Update} (Section~\ref{sec:update}).


\section{Choosing Seeds}
\label{sec:choose}

We now study two approaches for selecting $k$ seed nodes in the {\tt Choose}
function of Algorithm~\ref{fm}: heuristic-based
(Section~\ref{subsec:heuristics}) and explore-exploit strategies
(Section~\ref{subsec:ee}).

\subsection{Heuristic-Based Strategies}
\label{subsec:heuristics}

We first discuss two simple ways for choosing seeds from the uncertain influence graph $G$.

{\bf 1. \texttt{Random}.} This heuristic, which arbitrarily selects $k$ seed nodes, is based on the \emph{fairness} principle, where every user has the same chance to be activated.

{\bf 2. \texttt{MaxDegree}.} Given a node $p$ in $G$, we define the {\it out-degree} of~$p$ to be the number of outgoing edges of $p$ with non-zero influence probabilities. The {\it out-degree} of $p$ can mean the number of friends of the social network user represented by $p$, or their number of followers. Intuitively, if $p$ has a higher out-degree, it has a higher chance of influencing other users.  The \texttt{MaxDegree} heuristic simply chooses the nodes with $k$ highest out-degree values.

The main advantage of these two heuristics is that they are easy to implement.
However, they do not make use of influence probability information effectively.
In a social network, some users might be more influential than others. It may
thus be better to target users with higher influence probabilities on their
outgoing edges. The above heuristics also do not consider the feedback
information received from the activated users, which can be useful to obtain the
true values of the influence probabilities. We will examine a better
seed-selection method next.

\subsection{Explore-Exploit Strategies}
\label{subsec:ee}

The {\it Explore-Exploit} (EE) strategy chooses seed nodes based on influence
probabilities. Its main idea is to {\it exploit}, or execute an offline IM
algorithm, based on the influence information currently available. Since this
information may be uncertain, the seed nodes suggested by {\it exploit} may
not be the best ones. We alleviate this problem by using {\it explore}
operations, in order to improve the knowledge about influence probabilities.
Solutions for effectively controlling explore and exploit operations have been
studied in the {\it multi-armed bandit} (MAB) literature~\cite{Robbins:1952,rl}.
These MAB solutions inspire our development of the two seed-selection
strategies, namely {\it $\epsilon$-greedy} and {\it Confidence-Bound (CB)}.
Next, we present these two solutions in detail.

\paragraph*{\bf 1. $\epsilon$-greedy}

In this strategy (Algorithm~\ref{alg:ee}), a parameter $\epsilon$ is used
to control when to explore and when to exploit.  Specifically, with probability $1-\epsilon$, exploitation is carried out; otherwise, exploration is performed.

\begin{algorithm}
\caption{{\tt $\epsilon$-greedy$(G, k)$}\label{alg:ee}}
\begin{algorithmic}[1]
    \State{\bf Input:} uncertain influence graph $G=(V, E, P)$, budget $k$
    \State{\bf Output:} seed nodes $S$ with $|S|=k$

    \State sample $z$ from $\mathit{Bernoulli}(\epsilon)$
    \If {$z = 0$}
    	 $S \gets \mathtt{Explore}(G, k)$
    \Else
    	 $ \space \textit{  } S \gets \mathtt{Exploit}(G, k)$
    \EndIf		

    \State \Return $S$
\end{algorithmic}
\end{algorithm}

In {\tt Exploit}, we execute an offline IM algorithm, given the graph information we have obtained so far.
Recall that we model the influence probability $p_{ij}$ between nodes $i$
and $j$ as a probability distribution $P_{ij}$. We use the mean of $P_{ij}$ to represent $p_{ij}$, i.e.,
\begin{equation}
p_{ij} =  \mathbb{E}[P_{ij}] = \frac{\alpha_{ij}}{\alpha_{ij} + \beta_{ij}}. \nonumber
\end{equation}
A graph with the same node structure but with the $p_{ij}$ values on edges
constitutes an influence graph $G'$, on which the offline IM algorithm is executed.
Notice that when $\epsilon=0$, the solution reduces to exploit-only, i.e.,
the IM algorithm is run on $G'$ only.

The main problem of {\tt Exploit} is that estimating $p_{ij}$
by $\mathbb{E}[P_{ij}]$ can be erroneous. For example, when
$P_{ij}$ is a highly uncertain Beta distribution (e.g., the uniform distribution,
$\mathrm{B}(1,1)$), any value in $[0,1]$ can be the real influence probability. Let us consider a node~$i$ that has, in reality, a high
influence probability $p_{ij}$ on another node~$j$. Due to the large variance in $P_{ij}$, its value is underestimated. This reduces the chance that {\tt Exploit} chooses node $i$ to activate; consequently, the seed nodes selected may not be the best. The {\tt Explore} routine is designed to alleviate this problem. Rather than equating $p_{ij}$ to $\mathbb{E}[P_{ij}]$, $p_{ij}$ is over-estimated by using $P_{ij}$'s standard deviation, or $\sigma_{ij}$:
\begin{eqnarray}
p_{ij}&=&
\mathbb{E}[P_{ij}] + \sigma_{ij} \nonumber \\
&=&\frac{1}{\alpha_{ij}+\beta_{ij}} \left(\alpha_{ij} +
\sqrt{\frac{\alpha_{ij}\beta_{ij}}{\alpha_{ij} + \beta_{ij} + 1}}\right).
\nonumber
\end{eqnarray}
Then an offline IM algorithm on these new values of
$p_{ij}$ is performed.  A node $i$ that has a small chance to be chosen may now have a
higher probability to be selected.  Our experiments show that the use of {\tt Explore} is especially useful during the first few trials of the OIM solution, since the influence probability values during that time may not be very accurate.  From the feedback of activated users, we can learn more about the influence probabilities of the edges of $i$. We will discuss this in detail in Section~\ref{sec:update}.

This $\epsilon$-greedy algorithm has two problems. First, it is difficult to set
an appropriate $\epsilon$, which may have a large impact on its effectiveness.
Second, increasing $p_{ij}$ by $\sigma_{ij}$  may not always be good.
Based on these observations, we next propose an improved version of $\epsilon$-greedy.

\paragraph*{\bf 2. Confidence-Bound (\texttt{CB})}
The main idea of this strategy is to use a real-valued parameter $\theta$ to control the value of $p_{ij}$:
\begin{equation}
\label{eqn:cb}
    p_{ij}=\mathbb{E}[P_{ij}] + \theta \sigma_{ij}. 
\end{equation}
As shown in Algorithm~\ref{alg:cb}, for every edge $e$ from node $i$ to $j$, we compute its mean $\mu_{ij}$, variance $\sigma_{ij}$, and influence probability $p_{ij}$ based on $\theta$ (Lines 3-6). An offline IM algorithm is then run on $G'$, the influence graph with the probabilities computed by Equation~\ref{eqn:cb} (Lines 7-8). The set $S$ of seed nodes is then returned (Line 9).

\begin{algorithm}
\caption{{\tt CB($G, k$)}\label{alg:cb}}
\begin{algorithmic}[1]
    \State{\bf Input:} uncertain influence graph $G=(V, E, P)$, budget $k$
    \State{\bf Output:} seed nodes $S$ with $|S| = k$

    \For{$e \in E$}
       \State $\mu_{ij} \gets \frac{\alpha_{ij}}{\alpha_{ij} + \beta_{ij}}$
       \State $\sigma_{ij} \gets \frac{1}{(\alpha_{ij} + \beta_{ij})} \cdot \sqrt{\frac{\alpha_{ij}\beta_{ij}}{(\alpha_{ij} + \beta_{ij} + 1)}}$
	   \State $p_{ij} \gets \mu_{ij} + \theta \sigma_{ij}$
	\EndFor

    \State $G'\gets G$, with edge probabilities $p_{ij}, \forall(i,j)\in E$
	  \State $ S \gets \mathtt{IM}(G', k)$
	
    \State \Return $S$
\end{algorithmic}
\end{algorithm}

{\bf Setting $\theta$.}  The key issue of Algorithm~\ref{alg:cb} is how to
determine the value of $\theta$, so that the best $S$ can be found. Observe that
when $\theta=0$, $p_{ij}$ becomes $\mu_{ij}$ or $\mathbb{E}[P_{ij}]$, and {\tt
CB} reduces to {\tt Exploit} of the $\epsilon$-greedy algorithm. On the other
hand, when $\theta=1$,  $p_{ij}$ becomes $\mathbb{E}[P_{ij}]+\sigma_{ij}$, and
{\tt CB} is essentially {\tt Explore}. Thus, $\epsilon$-greedy is a special case
of {\tt CB}. However, {\tt CB} does not restrict the value of $\theta$ to zero
or one. Thus, {\tt CB} is more flexible and general than $\epsilon$-greedy.

In general, when $\theta>0$ is used, it means that {\tt CB} considers the
influence probabilities given by $\mu_{ij}$'s to be under-estimated, and it
attempts to improve the activation effect by using larger values of $p_{ij}$.
On the contrary, if $\theta < 0$, the influence probabilities are considered to
be over-estimated, and {\tt CB} reduces their values accordingly. As we will
discuss in Section~\ref{subsec:theta}, $\theta$ can be automatically adjusted
based on the feedback returned by activated users. This is better than
$\epsilon$-greedy, where the value of $\epsilon$ is hard to set. Note that we
choose to use a global $\theta$ instead of a local one on each edge, to reduce
the number of parameters to be optimized and to improve efficiency.


\section{Managing User Feedback}\label{sec:update}

Recall from Algorithm~\ref{fm} that after the seed nodes $S$ are obtained from
{\tt Choose} (Line 5), they are activated in the real world. We then collect
{\it feedback} from the users represented by these nodes (Lines~6--7). The
feedback describes which users are influenced, and whether each activation is
successful. For instances of such feedback traces, take for example Twitter and
other micro-bloggin platforms. In these, the system can track actions such as
likes and retweets which are reasonable indicators of influence propagation.  We
now explain how to use the feedback information to perform {\tt Update} (Line
8), which refreshes the values of influence probabilities and $\theta$ used in
the \texttt{CB} algorithm.

\begin{figure}[t!]
  \centering
  \includegraphics[width= .85 \linewidth]{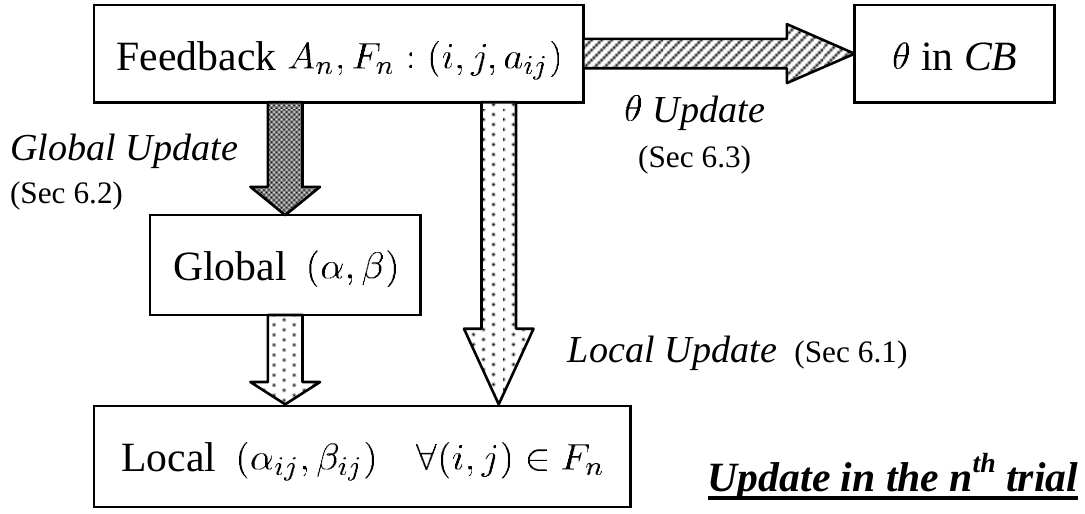}\\
  \caption{Updating the influence graph and $\theta$ with user feedback.}\label{fig:update}
\end{figure}

Given a trial $n$ in Algorithm~\ref{fm}, let $A_n$ be the set of activated nodes in that trial, and $F_n$ be the set of activation results. Specifically, $F_n$ contains tuples in the form of $(i,j, a_{ij})$, where $i$ and $j$ are users between which an activation was attempted; $a_{ij}=1$ if the
influence was successful, and $a_{ij}=0$ otherwise.  Note that $(i,j)$ is an edge of the influence graph $G$. Also, $F_n$ might not contain all edges of $G$, since an activation might not reach every user in $G$.

Three kinds of updates can be performed based on $A_n$ and  $F_n$:
\\{\bf 1. Local (Section~\ref{subsec:local}):} Update the influence probability's distribution (i.e., $\mathrm{B}(\alpha_{ij}, \beta_{ij})$) if the edge (i.e., activation from $i$ to $j$) was attempted;
\\{\bf 2. Global (Section~\ref{subsec:global}):} Update the global prior information $\alpha$ and~$\beta$, which are shared by all edges of $G$; and
\\{\bf 3. $\theta$ (Section~\ref{subsec:theta}):} Update the value of $\theta$ used in \texttt{CB}, if it is used as a seed selection strategy in {\tt Choose}.

Figure~\ref{fig:update} illustrates these three kinds of updates in the $n$-th trial. In the next sections, we discuss how to conduct these updates in detail. We remark that these update methods do not affect {\tt Random} and {\tt MaxDegree}, since they do not use these updated values.

\subsection{Updating Local $\vec{\alpha}$ and $\vec{\beta}$}
\label{subsec:local}

As we mentioned before, the influence probability between any two adjacent nodes
$i$ and $j$ is modeled as a Beta distribution with parameters $\alpha_{ij}$ and
$\beta_{ij}$, denoted as $P_{ij}\sim\mathrm{B}(\alpha_{ij},\beta_{ij})$.
Since the Beta distribution is a conjugate prior for the Bernoulli distribution,
then, given feedback $(i, j, a_{ij})$ in $F_n$ (seen as a Bernoulli trial), we
can update the distribution as follows:

{\bf 1.} If $a_{ij}=1$, i.e., the activation from node $i$ to node $j$
  was successful: $P_{ij}\sim \mathrm{B}(\alpha_{ij} + 1,\beta_{ij})$;

{\bf 2.} If $a_{ij}=0$, i.e., the activation from node $i$ to node $j$ failed:
  $P_{ij}\sim \mathrm{B}(\alpha_{ij},\beta_{ij} + 1)$.

In the beginning, we have no prior information about the
distribution except the global $\alpha$ and $\beta$, i.e., $\alpha_{ij}=\alpha$
and $\beta_{ij}=\beta$.  After $n$ trials and activations, we have thus
collected $n$ pieces of feedback information.  Let $h_{ij}$ ($m_{ij}$) be the number
of successful  (failed) activations for edge $(i,j)$. We have
\begin{equation}
    \alpha_{ij} = \alpha + h_{ij},\quad \beta_{ij} = \beta + m_{ij}. \nonumber
\end{equation}
Hence, this local update is equivalent to maintaining a distribution
$\mathrm{B}(\alpha+h_{ij},\beta+m_{ij})$, i.e., the distributions on the
edges simply count the number of successful and failed activations passing
through that edge, smoothed by the prior $\mathrm{B}(\alpha,\beta)$.

Note that this update process corresponds exactly to the MLE approach taken by~\cite{Goyal:2010} to learn influence probabilities
from action logs, with a smoothing prior added. The important difference is
that~\cite{Goyal:2010} only conducts this estimation for edges where there is
evidence, i.e., local updates. If the evidence is sparse, this can lead to
a sub-optimal, and overfitting, influence graph. Global update of Beta
priors, which go beyond the local feedback, can yield a better influence graph.

\subsection{Updating Global $\vec{\alpha}$ and $\vec{\beta}$}
\label{subsec:global}

Local updates to the random variable $P_{ij}$ allows the edge influence probability distribution to be updated directly.
In the first few trials, however, the real influence spread is sparse and limited,
and most of the edges will not be reached by an activation.
Therefore, the influence of choosing a good prior will weigh heavily on how {\tt
Choose} performs.  Once some evidence is gathered, this prior can be refined by
taking into account the feedback in a \emph{global} sense, over all trials up to
the current one. Next, we present two methods of updating the global
$\alpha$ and $\beta$ priors based on the feedback.

\paragraph*{Least Squares Estimation} The first solution is to find the best
fit for the $\alpha$ and $\beta$ priors according to the real spread that we obtained from the real world test at each trial.

Let us first explain the reasoning when there is one seed node (i.e., $|S_n|=1$), and we fix $\alpha=1$.
Let $\mathcal{A}_n$ be the set of successful activated nodes before the $n$-th trial (i.e., $\mathcal{A}_n = \cup_{l=1}^{n - 1}A_l$), and $\sigma_n(\{i\})$ be the expected number of additional activated nodes (or expected additional spread) from the seed node $i$ in the $n$-th trial.
For $S_n = \{s\}$, $\sigma_n(\{s\})$ is:
\begin{equation}
{\small
    \sigma_n(\{s\}) = 1 + \sum_{\substack{(s,i)\in E\\ i\not\in \mathcal{A}_n}}{p_{si}\times\sigma_n(\{i\})}
    + \sum_{\substack{(s,i)\in E\\ i\in \mathcal{A}_n}}{p_{si}\times(\sigma_n(\{i\}) - 1)}, \nonumber
}
\end{equation}
which is the sum of the outgoing spreads weighted by the outgoing
probabilities $p_{si}$ and discounted by 1 for
nodes already activated along an outgoing edge.

We estimate $\sigma_n(\{s\})$ by $|A_n|$ from the feedback obtained by the
influence campaign. We also estimate $p_{si}=\frac{\alpha+h_{si}}{\alpha +
  h_{si} + \beta + m_{si}}$, i.e., the mean of $\mathrm{B}(\alpha_{ij},
  \beta_{ij})$. Note that $h_{si}+m_{si}$ is the total number of attempts from
  node $s$ to $i$, which is the same for neighbors of $s$ because every
  activation through $s$ tries to activate all outgoing nodes in the independent
  cascade model. Thus, we use $t_s$ to denote $h_{si}+m_{si} \ \forall
  (s,i)\in E$.  By further estimating $\sigma_n(\{i\})$ by an overall estimation
  $\hat{\sigma}_n$ and set $\alpha=1$, we obtain
{
\begin{equation}
|A_n| = 1 + \frac{1}{\beta + t_s + 1}\left(\sum_{(s,i)\in E}{(h_{si}+1)\hat{\sigma}_n} - \sum_{(s,i)\in E, i\in\mathcal{A}_n}{(h_{si}+1)}\right).
\nonumber
\end{equation}}

Let $o_s$ be the outgoing degree of $s$, $a_s$ be the number of (previously)
activated neighbors of $s$ (i.e., $a_s=|\{i|(s,i)\in E \wedge
  i\in\mathcal{A}_n\}|$), $h_s$ be the number of total successful activations
  (or hits) on outgoing edges of $s$, and $h_{as}$ be the number of total hits
  on edges leading to activated neighbors. The above equation is simplified to
\begin{equation}
(|A_n| - 1)\beta = (1-|A_n|)(t_s+1) +(h_s+o_s)\hat{\sigma}_n - (h_{as}+a_s). \nonumber
\end{equation}

We then rewrite it as the form of $x_n \beta = y_n$. Since this equation also
applies to activations in all trials up to the current one, we use the least
square estimator for linear regression without an intercept term to get an
estimator for
$\beta$, $\hat{\beta}$:
\begin{equation}
    \hat{\beta} = (\vec{x} \cdot \vec{y})\ /\ (\vec{x} \cdot \vec{x}), \nonumber
\end{equation}
where $\vec{x}$ and $\vec{y}$ are the vectors of values $x_n$ and $y_n$.
The same principles apply when estimating
$\alpha$ and $\beta$ simultaneously, and we omit the details here.

We estimate $\hat{\sigma}_n$ by the average spread of the node from the activation
campaigns, i.e., $\hat{\sigma}_n = \sum_{l=1}^{n}|A_n| /
\sum_{l=1}^{n}|S_n|$.  Note that, when $\hat{\sigma}_n=0$, the equation for
$|A_n|$ is exactly the degree discount estimator from the IM literature~\cite{wei09},
and represents a lower bound on the spread from a node.

A further complication occurs when $|S_n|>1$, which might result in an equation at
least quadratic in $\beta$, due to the influence probability equations of nodes
which are neighbors of more than one seed node. In this work, we simplify the
estimation by assuming full independence among seed nodes, and hence replacing
$x_n$ and $y_n$ by the sum over all $s\in S_n$.

We remark that the estimator above suffers from the reliance on the spread
estimation $\hat{\sigma}_n$. However, it is a good option when we cannot access the full activation feedback $F_n$, but instead, do have
the access to the set of successful activated nodes in each trial (i.e.,
the set~$A_n$). This may happen in an alternate problem setting
when one cannot get all the feedback information from the activated users in~$A_n$.

\paragraph*{Maximum Likelihood Estimation} Given the feedback from each trial
$n$, we can compute the likelihood of the feedback $F_n$ given the probabilities
of each edge in the feedback tuples, by assuming they are activated
independently. 
The likelihood depends on the successful activations (hits) and failed activations (misses)
of each edges and the global prior parameters $\alpha$ and
$\beta$:
\begin{eqnarray}
 \mathcal{L}(F_n)&=&\prod_{(i,j,a_{ij})\in F_n} p_{ij}^{a_{ij}}(1-p_{ij})^{1-a_{ij}}, \nonumber\\
 \mathcal{L}(F_n\mid\alpha,\beta)&=&\prod_{(i,j,a_{ij})\in F_n}
 \frac{(\alpha+h_{ij})^{a_{ij}}(\beta+m_{ij})^{1-a_{ij}}}{\alpha+\beta+h_{ij}+m_{ij}}. \nonumber
\end{eqnarray}

We need to find the parameters $\alpha$ and $\beta$ which maximize the
likelihood:
\begin{equation}
\argmax_{\alpha,\beta} \mathcal{L}(F_n\mid\alpha,\beta).  \nonumber
\end{equation}

To simplify calculations we take the maximum of the log likelihood:
\begin{align*}
  \log\mathcal{L}(F_n\mid\alpha,\beta)&=\sum_{(i,j,a_{ij})\in
  F_n}a_{ij}\log(\alpha+h_{ij})\\&\quad+\sum_{(i,j,a_{ij})\in
  F_n}(1-a_{ij})\log(\beta+m_{ij})\\&\quad-\sum_{(i,j,a_{ij})\in
  F_n}\log(\alpha+\beta+h_{ij}+m_{ij}).
\end{align*}

The optimal values are obtained by solving the equations $\frac{\partial
  \log\mathcal{L}(F_n\mid\alpha,\beta)}{\partial \alpha}=0$ and $\frac{\partial
  \log\mathcal{L}(F_n\mid\alpha,\beta)}{\partial \beta}=0$ for $\alpha$ and
  $\beta$, respectively, which can be simplified as
%
%
\begin{equation}
\sum_{(i,j,a_{ij})\in F_n,a_{ij}=1} \frac{1}{\alpha+h_{ij}} =
\sum_{(i,j,a_{ij})\in F_n,a_{ij}=0}\frac{1}{\beta+m_{ij}}
\nonumber
\end{equation}

This equation can be solved numerically by setting $\alpha$ and solving $\beta$.
In practice, we can fix $\alpha=1$, and let $f(\beta)$ be
\begin{equation}
    f(\beta) = \sum_{(i,j,a_{ij})\in F_n,a_{ij}=0}\frac{1}{\beta+m_{ij}} - \sum_{(i,j,a_{ij})\in F_n,a_{ij}=1} \frac{1}{\alpha+h_{ij}}
    \nonumber
\end{equation}
Since $f(\beta)$ is a monotonically decreasing function of $\beta$ ($f'(\beta) \leq 0\ \forall\beta\geq 1$), we can solve $f(\beta) = 0$ by a binary search algorithm with an error bound $\eta$ (e.g., $10^{-6}$). And thus, the global $\alpha$ and $\beta$ priors are updated accordingly.
%
%
%
%


\subsection{Updating $\vec{\theta} $}\label{sec:eg}
\label{subsec:theta}

We now explain how to dynamically update the value of $\theta$ used in the \texttt{CB} strategy
(Section~\ref{subsec:ee}).

Let $\vec{\theta}=\{\theta_1,\theta_2,\dots,\theta_q\}$ be the $q$ possible values of $\theta$.
We also let $\vec{\phi}=\{\phi_1,\phi_2,\dots,\phi_q\}$, where $\phi_j$ is the probability of using $\theta_j$ in \texttt{CB}.  Initially, $\phi_j = 1/q$ for $j=1,\ldots,q$, and its value is updated based on the gain obtained in each trial. The gain is defined as $G_n=|A_n| / |V|$,
where $|A_n|$ is the real influence spread observed in each round. We then determine $\vec{\theta}$ by using the exponentiated gradient algorithm~\cite{Cesa}. The rationale of using this solution is that if the value of $\theta_j$ used in this trial results in a high gain, the
corresponding $\phi_j$ will be increased by the algorithm, making $\theta_j$ more likely to be chosen in the next trial. Algorithm~\ref{alg:eg} gives the details.

\begin{algorithm}
  \caption{{\tt
  ExponentiatedGradient}($\vec{\phi},\delta,G_n,j,\mathbf{w}$)\label{alg:eg}}

{
\begin{algorithmic}[1]
  \State{\bf Input:} $\vec{\phi}$, probability distribution;
$\delta$, accuracy parameter; $G_n$, the gain obtained;  $j$, the index
of latest used $\theta_j$; $\mathbf{w}$, a vector of weights; $N$, the number of trials.
  \State{\bf Output:} $\theta$
  \State $\gamma \gets \sqrt{\frac{\ln(q/\delta)}{q N}}$, $\tau \gets \frac{4  q
  \gamma}{3 + \gamma}$, $\lambda \gets \frac{\tau}{2 q}$
  \For{$i=1$ \textbf{to} $q$}
    \State $w_i \gets w_i \times \exp\left(\lambda \times \frac{G_n\times
\mathbb I[i=j]
    + \gamma}{\phi_i}\right)$
  \EndFor
  \For{$i=1$ \textbf{to} $q$}
    \State $\phi_i \gets (1 - \tau) \times \frac{w_i}{\sum_{j=1}^{k} w_j} + \tau
    \times \frac{1}{q}$
	\EndFor
  \State{\bf return} sample from $\vec{\theta}$ according to $\vec{\phi}$ distribution
\end{algorithmic}
}
\end{algorithm}

Here, $\gamma$ and $\lambda$ are smoothing factors used to update weights,
and $\mathbb I[z]$ is the indicator function. We compute
$\vec{\phi}$ by normalizing vector $\mathbf{w}$ with regularization factor
$\tau$. All the values in $\mathbf{w}$ are initialized with the value of 1.

In~\cite{Cesa}, it is shown that, for a choice of constant $\theta$'s, {\tt
ExponentiatedGradient} can provide a regret bound on the optimal sequence of
chosen $\theta$ in the vector. In our case, the experimental results also show
that {\tt ExponentiatedGradient} is the best performing strategy.


\begin{figure*}[t!]
  \centering
  \includegraphics[width= \linewidth]{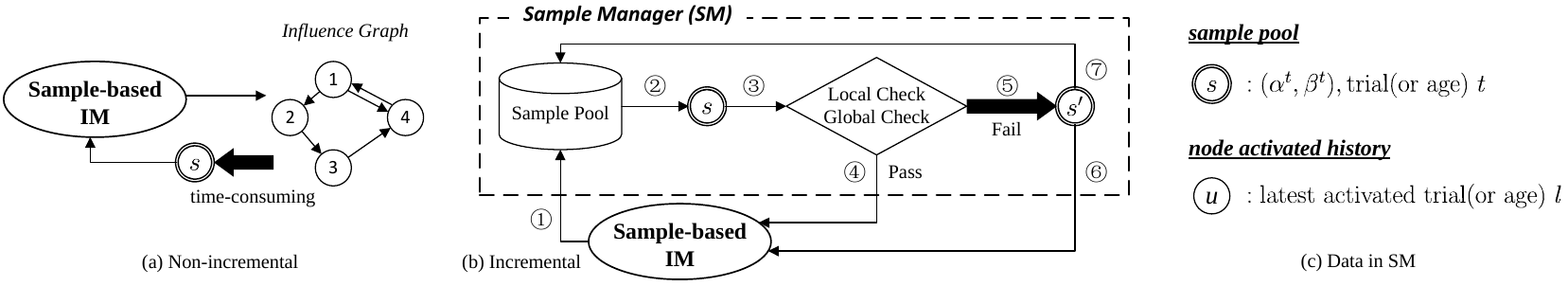}\\
  \caption{Sample-based IM algorithms.}\label{fig:inc}
  \vspace{-10pt}
\end{figure*}

\section{Incremental Solution for OIM}
\label{sec:inc}

In our OIM framework, an IM algorithm is invoked once in every trial to select
seeds.  However, the state-of-the-art IM algorithms with good theoretical
approximation bounds, such as CELF, TIM, and TIM+, are generally costly to run,
especially for large graphs with high influence probabilities. For instance, in
our experiments in the DBLP dataset\footnote{{Detailed description of the
dataset is given in Section~\ref{sec:experiments}.}}, which has around
2{,}000{,}000 edges, the best known algorithm (TIM+) also takes around
half an hour to select the nodes for a trial. Since every run of OIM
takes multiple trials, the running time can be too high in practical terms. To
alleviate this issue, we explore in this section the possibility to increase
the scalability of the OIM framework, by re-using computations between
trials.


The first observation is that all the IM algorithms with theoretical approximation bounds are sample based, and
follow the general sampling process illustrated in Figure~\ref{fig:inc}(a).
Every time an algorithm requires a sample, it samples the influence graph based
on the edge influence probabilities and stores it in a sample, say $s$.
Moreover, their running time is dominated by the cost of sampling the influence graph
(the thick arrow in Figure~\ref{fig:inc}(a)).  For example, more than 99\% of
the computation of TIM+ is spent in sampling the random reverse reachable sets
in the influence graph~\cite{TIM}.

Secondly, the size of the real-world feedback $F_n$ is relatively small compared
with the number of edges in a graph. For instance, in DBLP with $k=1$ and
using TIM+, the average $|F_n|$ is less than $1\%$ of the total number of edges
in the graph. This makes intuitive sense.
Since samples are generated based on the influence graph, and the real-world feedback only influences a small part of
the graph, it would only affect few samples taken from the updated influence graph in the next trial.
This motivates us to explore methods which can save the computational
effort, especially the effort in sampling,
by reusing samples of previous trials,
without incurring much error. 

\subsection{Solution Framework}
\label{subsec:inc-framework}

To explain our approach, we introduce a sample manager (SM) which is responsible
for the sampling procedure for the sample-based IM algorithms.  Generally
speaking, when the IM algorithm requires a sample of the influence graph, it
sends the request to SM, which will then return a sample to it.  To enable an
incremental approach that reuses the computational effort, SM stores the samples
from the previous iterations in a sample pool. In the new trial, it attempts to
reuse the stored samples, if possible, instead of sampling the influence graph
again.

The principle of SM is illustrated in Figure~\ref{fig:inc}(b).  In a new trial,
when the sample-based IM algorithm requires a sample, it sends requests to SM
({\bf Step 1}). SM then randomly selects a sample $s$, which has not been used in this
trial, from the sample pool ({\bf Step 2}). After that, SM conducts two checks, called
local check and global check, on $s$, whose purpose is to determine whether $s$
is allowed to be reused after local and global updates performed in previous
rounds ({\bf Step 3}). If $s$ passes these two checks, SM simply returns the sample to
the IM algorithm ({\bf Step 4}); otherwise, SM generates a new sample $s'$ based on
the current influence graph ({\bf Step 5}), and returns it to the IM algorithm ({\bf
Step 6}) as well as replaces $s$ by $s'$ in the sample pool ({\bf Step 7}).

In the above framework, assuming that conducting the local and global checks is
much more efficient than sampling the influence graph and the ratio of reused
samples is high, SM has the potential to significantly reduce the running time
of the IM algorithm in the OIM framework.

Next, we demonstrate how this principle can be applied in practice on the TIM+
algorithm.
Please note that the same principle can be easily applied to develop the incremental approaches
for other sample-based IM algorithms.

\subsection{Case Study: TIM+}
\label{subsec:tim}

In this section, we demonstrate the case that TIM+ is executed when an IM algorithm is called in OIM framework.
For example, in {\tt Explore}, TIM+ is run with the input influence graph obtained by taking the mean of the random variable as the influence probability of the edge, i.e., $p_{ij} = \frac{\alpha_{ij}}{\alpha_{ij} + \beta_{ij}}$.
We next demonstrate how to develop the incremental approach for TIM+ in {\tt
Exploit} with SM. The principle also applies for {\tt Explore} as well as {\tt
CB}. We focus on {\tt Exploit} here and omit details for others.

Briefly speaking, TIM+ generates a set of {\it random reverse reachable sets}
(or random RR sets) on the influence graph, and estimates the expected spread of
nodes, based on the generated random RR sets.  Here, an RR set for node $v \in
V$, denoted by $R_v$, is a set of nodes which are: (1) generated on an instance
of a randomly sampled influence graph $g$ (an edge exists with a probability
equal to its influence probability), and (2) able to reach $v$ in the sampled
graph $g$.  In other words, $\forall i \in R_v$, there exists a path from $i$ to
$v$ in $g$.  A random RR set is then an RR set where $v$ is selected uniformly
at random from $V$. We omit the formal definition of
random RR sets as well as their generation and refer interested readers
to~\cite{TIM} for details.

Let $E(R_v)$ be the set of all incoming edges for nodes in $R_v$, i.e., $E(R_v) = \{(i, j)|(i, j)\in E \wedge j\in R_v\}$. The next lemma is the foundation of the incremental approach for TIM+.

\begin{lemma}
\label{lemma:inc}
Given node $v\in V$, the occurrence probability of an RR set ($R_v$) keeps unchanged if the influence probabilities for edges in $E(R_v)$ do not change.
\end{lemma}
\begin{proof}
Let $\xi_{ij}$ be a random variable for the existence of edge $(i,j)\in E$. We have $\xi_{ij}=1$ with probability of $p_{ij}$, and $\xi_{ij}=0$ with probability of $1-p_{ij}$. $\Pr(R_v)$ is the probability that the following two events happen:
(i) $\forall i \in R_v$, there exists a path from $i$ to $v$, i.e., $\exists j\in R_v$ s.t. $\xi_{ij}=1$; and
(ii) $\forall i \not\in R_v$, there exists no path from $i$ to $v$, i.e., $\forall j\in R_v$, $\xi_{ij}=0$.
And therefore, $\Pr(R_v)$ is some function of $p_{ij}$ where $(i,j)\in E(R_v)$. Hence, if $p_{ij}$ $(\forall (i,j)\in E(R_v))$ is unchanged, $\Pr(R_v)$ keeps unchanged, too.
\end{proof}

Let us consider SM introduced in Section~\ref{subsec:inc-framework}. The samples
stored in SM for TIM+ are the random RR sets described above. After each
round, local and global graph updates may be performed according to the real
world feedback.
Suppose the current trial is $n$, for a randomly selected $s$ (or $R_v$) from the sample pool, we have to conduct local and global checks for it.
Lemma~\ref{lemma:inc} gives an intuition on how these checks can be performed for these checks. $\Pr(R_v)$ remains the same (or only deviates a bit) if the updates have no effect (or only some minor effects) on the influence probabilities for edges in $E(R_v)$.

Before we detail the local and global checks,
let us first define the age of a sample $s$ and the age of a node $u$.
The age of $s$ is the trial when $s$ was sampled, and the age of a node $u$ is
the latest trial when the real world test attempted to activate $u$ (regardless
of the activation's success).

{\bf Local check.}
Let $R_v$'s age to be $t$, and $E_{local}$ be the set of edges that exist in the feedbacks from the $t$-th trial to the $(n-1)$-th trial, i.e., $E_{local} = \{(i,j)|\exists q (t\leq q\leq n - 1) \text{ s.t. } (i,j) \in F_q\}$.
Local updates only affect edges that are included in the real world feedback, and so, $E_{local} \cap E(R_v) = \emptyset$ indicates that influence probabilities for edges in $E(R_v)$ in the $n$-th trial are the same as the ones in the $t$-th trial. Hence, $R_v$ is not affected by local updates. In other words,
\[
   \left( E_{local} \cap E(R_v) = \emptyset \right)
   \Rightarrow
   \left( \text{$R_v$ passes local check} \right).
\]

We use the sample and node ages for an efficient local check as follows.
\begin{lemma}[Local Check]
    \label{lemma:local}
    If for all $u \in R_v$, $u$'s age is smaller than $R_v$'s age, we have $E_{local} \cap E(R_v) = \emptyset$.
\end{lemma}
\begin{proof}
Recall that $u$'s age, denoted $l$, is the latest trial that the real world test tried to activate it. We have,
\\$(l < t) \Rightarrow (\forall (i,u)\in E, (i,u,a_{iu})\not\in F_q (t\leq q\leq n-1))$. Lemma~\ref{lemma:local} is then a direct consequence by considering the definition of $E_{local}$.
\end{proof}

According to Lemma~\ref{lemma:local}, we store the sample as well as its age in the sample pool, and we also store the node's age in a node activated history (refer to Figure~\ref{fig:inc}(c)). Then, the time complexity to do local check is $O(|R_v|)$ as the age's information can be accessed in constant time.

{\bf Global check.}
After global update is performed, the global $\alpha$ and $\beta$ priors may be
changed. Since they are shared by all edges, changes on global priors lead to
changes on all edges' influence probabilities. However, we observe that they
will converge as we get more activation feedback from the real world. Intuitively, if the influence probabilities for edges in $E(R_v)$ only deviate a bit, there is only minor effect on the random RR sets. Note that, only samples which pass local check will be then evaluated by the global check. And so, if the global priors when the sample $s$ (or $R_v$) was generated are close to the current global priors, the influence probabilities for edges in $E(R_v)$ do not change much.

Let $\alpha^t$ and $\beta^t$ be the priors at trial $t$, and the current priors are $\alpha$ and $\beta$. We use a threshold $\tau$ to measure whether two priors are close, moreover, whether global check is passed.
\[
   \left( \left| \frac{\alpha^t}{\alpha^t+\beta^t} - \frac{\alpha}{\alpha+\beta} \right| < \tau \right)
   \Rightarrow
   \left( \text{$R_v$ passes global check} \right).
\]

Hence, in SM, we also store the priors when the sample was generated in the sample pool (Figure~\ref{fig:inc}(c)). And therefore, the global check is conducted in constant time $O(1)$.


{\bf Discussions.} The total time complexity of conducting local and global
checks on a sample $R_v$ is $O(|R_v|)$. As mentioned in~\cite{TIM}, the
complexity of generating a sample $R_v$ is of the order of the total in-degree for nodes in $R_v$, i.e., $O(|E(R_v)|)$. Let $d$ be the average in-degree for a node, we have $|E(R_v)| = d\times|R_v|$ on average. This indicates that conducting checks for a sample is about $d$ times faster than generating a new sample. Hence, the incremental approach for TIM+ with SM can significantly save computation effort if the ratio of re-used samples is high.

Note that if {\tt CB} is employed, $\theta$ may also be updated according to the real-world feedback. We design a similar mechanism with global check, called $\theta$ check, to verify whether $R_v$ is allowed to be re-use. Let $\theta^t$ be the $\theta$ when $R_v$ was generated and $\sigma^t$ be the standard deviation for global prior. We have
\[
   \left( \left| \theta^t \sigma^t - \theta \sigma \right| < \tau \right)
   \Rightarrow
   \left( \text{$R_v$ passes $\theta$ check} \right).
\]
In the next section, we show our experimental results to verify our OIM framework.


\newcommand{\plreal}[7]{\addplot[black, mark=#6, each nth point = 4, mark size = 3] table[x=t,y=#5] {data/#1.#2.#3.#4.dat}}

\newcommand{\plrealx}[7]{\addplot[black, mark=#6,#7, each nth point = 4, mark size = 3] table[x=t,y=#5] {data/#1.#2.#3.#4.dat}}
\newcommand{\plrealy}[7]{\addplot[black, mark=#6, mark size = 3] table[x=t,y=#3] {exp/f/#1_#2_#4.dat}}
\newcommand{\plrealz}[7]{\addplot[black, mark=#6, mark size = 1] table[x=t,y=#5] {#1/#1_#2_K#3.dat}}

\newcommand{\plee}[7]{\addplot[black, mark=#6, each nth point = 5, mark size = 3] table[x=t,y=#5] {#1/#1_#3_K#2#4.dat}}
\newcommand{\plkt}[7]{\addplot[black, mark=#6, each nth point = 5, mark size = 3] table[x=t,y=#5] {exp/#1_#2_K#3T#4.dat}}

\newcommand{\plotopt}[3]{
  \begin{subfigure}{\linewidth}
    \centering
    \begin{tikzpicture}[scale=0.44]
    \tikzstyle{every node}=[font=\Large]
    \pgfsetplotmarksize{1.2pt}
    \begin{axis}[%
        enlargelimits=0,
        ymin = 700,
        ymax = 1200,
        axis on top,
        xlabel={k (\textsc{#1}, Budget = 50)},
        ylabel={Influence Spread},
        legend columns = 3,
        legend entries = {\texttt{MaxDegree}, \texttt{CB}, \texttt{CB-INC}},
        legend to name = legopt,
        cycle list name = {color}
      ]
      \plrealy{#1}{MaxD}{}{B}{re}{square}{solid};
      \plrealy{#1}{Online}{}{B}{re}{*}{solid};
      \plrealy{#1}{OnlineINC}{}{B}{re}{o}{solid};
    \end{axis}
  \end{tikzpicture}
 \end{subfigure}
}

\newcommand{\plottau}[4]{
  \begin{subfigure}{\linewidth}
    \centering
    \begin{tikzpicture}[scale=0.44]
    \tikzstyle{every node}=[font=\Large]
    \pgfsetplotmarksize{1.2pt}
    \begin{axis}[%
	xmin = 0, xmax = 0.04,
        ymin = 500, ymax = 1600,
        axis on top,
        xlabel={Tau, (\textsc{#1}, Budget = 50)},
        ylabel={#3}
      ]
      \plrealy{#1}{Tau}{re}{#4}{#2}{+}{solid};
    \end{axis}
  \end{tikzpicture}
 \end{subfigure}
}

\newcommand{\plotKT}[3]{
  \begin{subfigure}{0.3\linewidth}
    \centering
    \begin{tikzpicture}[scale=0.44]
    \tikzstyle{every node}=[font=\Large]
    \pgfsetplotmarksize{1.2pt}
    \begin{axis}[%
        enlargelimits=0,
        axis on top,
        xlabel={Trial (\textsc{#1}, k = #2)},
        ylabel={Influence Spread},
        legend columns = 5,
        legend entries = {\texttt{Real}, \texttt{Random}, \texttt{MaxDegree}, \texttt{CB}, \texttt{CB-INC}},
        legend to name = legkt,
        cycle list name = {color}
      ]
      \plkt{#1}{Real}{#2}{#3}{re}{diamond}{solid};
      \plkt{#1}{Random}{#2}{#3}{re}{x}{solid};
      \plkt{#1}{MaxD}{#2}{#3}{re}{square}{solid};
      \plkt{#1}{Online}{#2}{#3}{re}{*}{solid};
      \plkt{#1}{OnlineINC}{#2}{#3}{re}{o}{solid};
    \end{axis}
  \end{tikzpicture}
 \end{subfigure}
}

\newcommand{\plotKTsmall}[3]{
    \begin{tikzpicture}[scale=0.44]
    \tikzstyle{every node}=[font=\Large]
    \pgfsetplotmarksize{1.2pt}
    \begin{axis}[%
        enlargelimits=0,
        axis on top,
        xlabel={Trial (\textsc{#1}, k = #2)},
        ylabel={Influence Spread},
        legend columns = 5,
        legend entries = {\texttt{Real}, \texttt{Random}, \texttt{MaxDegree}, \texttt{CB}, \texttt{CB-INC}},
        legend to name = legkt,
        cycle list name = {color}
      ]
      \plkt{#1}{Real}{#2}{#3}{re}{diamond}{solid};
      \plkt{#1}{Random}{#2}{#3}{re}{x}{solid};
      \plkt{#1}{MaxD}{#2}{#3}{re}{square}{solid};
      \plkt{#1}{Online}{#2}{#3}{re}{*}{solid};
      \plkt{#1}{OnlineINC}{#2}{#3}{re}{o}{solid};
    \end{axis}
  \end{tikzpicture}
}

\newcommand{\plotee}[3]{
  \begin{subfigure}{0.45\linewidth}
    \centering
    \begin{tikzpicture}[scale=0.44]
    \tikzstyle{every node}=[font=\Large]
    \pgfsetplotmarksize{1.2pt}
    \begin{axis}[%
        xmin=0, xmax=50,
        enlargelimits=0,
        axis on top,
        xlabel={Trial (\textsc{#1}, k = #2)},
        ylabel={Influence Spread},
        legend columns = 3,
        legend entries = {\texttt{Exploit}, \texttt{$\epsilon$-greedy},
                          \texttt{CB}},
        legend to name = legee,
        cycle list name = {color}
      ]

      \plkt{#1}{MLE}{#2}{#3}{re}{square}{solid};
      \plkt{#1}{EE}{#2}{#3}{re}{triangle}{solid};
      \plkt{#1}{Online}{#2}{#3}{re}{*}{solid};

    \end{axis}
  \end{tikzpicture}
 \end{subfigure}
}

\newcommand{\plotupdates}[3]{
    \begin{tikzpicture}[scale=0.44]
    \tikzstyle{every node}=[font=\Large]
    \pgfsetplotmarksize{1.2pt}
    \begin{axis}[%
        xmin=0, xmax=50,
        enlargelimits=0,
        axis on top,
        xlabel={Trial (\textsc{#1}, k = #2)},
        ylabel={Influence Spread},
        legend columns = 4,
        legend entries = {\texttt{CB+MLE}, \texttt{CB+LSE}, \texttt{CB+LOC}, \texttt{CB+NO}},
        legend to name = legupdates,
        cycle list name = {color}
      ]

      \plkt{#1}{Online}{#2}{#3}{re}{*}{solid};
      \plkt{#1}{LSE}{#2}{#3}{re}{o}{solid};
      \plkt{#1}{LOC}{#2}{#3}{re}{square}{solid};
      \plkt{#1}{NO}{#2}{#3}{re}{+}{solid};

    \end{axis}
  \end{tikzpicture}
}

\newcommand{\plotpriors}[2]{
    \begin{tikzpicture}[scale=0.44]
    \tikzstyle{every node}=[font=\Large]
    \pgfsetplotmarksize{1.2pt}
    \begin{axis}[%
	xmin=0, xmax=220,
	ymin=500, ymax=1200,
        enlargelimits=0,
        axis on top,
        xlabel={Beta (\textsc{#1}, k = 1)},
        ylabel={Influence Spread},
        legend columns = 3,
        legend entries = {\texttt{CB+MLE}, \texttt{CB+LSE}, \texttt{CB+LOC}, \texttt{CB+No Update}},
        legend to name = legprios,
        cycle list name = {color}
      ]
      \plrealy{#1}{MLE}{re}{K#2}{}{*}{solid};
      \plrealy{#1}{LSE}{re}{K#2}{}{o}{solid};
      \plrealy{#1}{LOC}{re}{K#2}{}{square}{solid};
      \plrealy{#1}{NO}{re}{K#2}{}{+}{solid};

    \end{axis}
  \end{tikzpicture}
}

\newcommand{\plott}[3]{
  \begin{tikzpicture}[scale=0.44]
  \tikzstyle{every node}=[font=\Large]
    	\begin{axis}[%
        xmin=0, xmax=50,
        enlargelimits=0,
        axis on top,
        xlabel={Trial (\textsc{#1}, k = #2)},
        ylabel={Running Time (in seconds)},
        legend columns = 4,
        legend entries = {\texttt{Random}, \texttt{MaxDegree}, \texttt{CB},\texttt{CB-INC}},
        legend to name = legtt,
        cycle list name = {color}
      ]
      \plkt{#1}{Random}{#2}{#3}{te}{+}{solid};
      \plkt{#1}{MaxD}{#2}{#3}{te}{square}{solid};
      \plkt{#1}{Online}{#2}{#3}{te}{*}{solid};
      \plkt{#1}{OnlineINC}{#2}{#3}{te}{o}{solid};

      \end{axis}
 \end{tikzpicture}
}

\newcommand{\plottt}[1]{
  \begin{tikzpicture}[scale=0.44]
  \tikzstyle{every node}=[font=\Large]
    	\begin{axis}[%
        enlargelimits=0,
        axis on top,
	xmin = 0,
        xlabel={k (\textsc{#1}, Budget = 50)},
        ylabel={Running Time (in seconds)},
        legend columns = 4,
        legend entries = {\texttt{Random}, \texttt{MaxDegree}, \texttt{CB},\texttt{CB-INC}},
        legend to name = legtt,
        cycle list name = {color}
      ]
      \plrealy{#1}{Random}{}{T}{re}{+}{solid};
      \plrealy{#1}{MaxD}{}{T}{re}{square}{solid};
      \plrealy{#1}{CB}{}{T}{re}{*}{solid};
      \plrealy{#1}{CBI}{}{T}{re}{o}{solid};

      \end{axis}
 \end{tikzpicture}
}

\newcommand{\plottsmall}[3]{
    \begin{tikzpicture}[scale=0.44]
    \tikzstyle{every node}=[font=\Large]
    	\begin{axis}[%
        xmin=0, xmax=50,
        enlargelimits=0,
        axis on top,
        xlabel={Trial (\textsc{#1}, k = #2)},
        ylabel={Running Time (in seconds)},
        legend columns = 4,
        legend entries = {\texttt{Random},\texttt{MaxDegree}, \texttt{CB},\texttt{CB-INC}},
        legend to name = legtt,
        cycle list name = {color}
      ]
      \plkt{#1}{Random}{#2}{#3}{te}{+}{solid};
      \plkt{#1}{MaxD}{#2}{#3}{te}{square}{solid};
      \plkt{#1}{Online}{#2}{#3}{te}{*}{solid};
      \plkt{#1}{OnlineINC}{#2}{#3}{te}{o}{solid};

      \end{axis}
 \end{tikzpicture}
}

\newcommand{\plotexploit}[2]{
  \begin{subfigure}{0.3\linewidth}
    \centering
    \begin{tikzpicture}[scale=0.6]
    \tikzstyle{every node}=[font=\Large]
    \pgfsetplotmarksize{1.2pt}
    \begin{axis}[%
        xmin=1, xmax=50,
        enlargelimits=0,
        axis on top,
        xlabel={Trial -- \textsc{#1} k=#2 Beta(1, 19)},
        ylabel={Influence Spread},
        legend columns = 5,
        legend entries = {\texttt{Random},\texttt{MaxDegree},
                          \texttt{Local},
                          \texttt{LSE},\texttt{MLE}},
        legend to name = legexp,
        cycle list name = {color}
      ]
      \plreal{#1}{real}{#2}{1}{re}{asterisk}{dotted};
      \plreal{#1}{real}{#2}{2}{re}{diamond}{dotted};
      \plreal{#1}{oim}{#2}{1.9.0.0.0}{re}{triangle}{dashed};
      \plreal{#1}{ols}{#2}{1.19.0.0.0}{re}{square}{solid};
      \plreal{#1}{mle}{#2}{1.19.0.0.0}{re}{*}{solid};
    \end{axis}
  \end{tikzpicture}
 \end{subfigure}
}

\newcommand{\plottime}[2]{
 \begin{subfigure}{0.3\linewidth}
    \centering
    \begin{tikzpicture}[scale=0.6]
    \tikzstyle{every node}=[font=\Large]
    \pgfsetplotmarksize{1.2pt}
    \begin{semilogyaxis}[%
        xmin=1, xmax=50,
        enlargelimits=0,
        axis on top,
        xlabel={trial -- \textsc{#1} k=#2 Beta(1, 9)},
        ylabel={time (min)},
        cycle list name = {color}
      ]
      \plreal{#1}{real}{#2}{1}{te}{asterisk}{dotted};
      \plreal{#1}{real}{#2}{2}{te}{diamond}{dotted};
      \plreal{#1}{oim}{#2}{1.9.0.0.0}{te}{triangle}{dashed};
      \plreal{#1}{ols}{#2}{1.19.0.0.0}{te}{square}{solid};
      \plreal{#1}{mle}{#2}{1.19.0.0.0}{te}{*}{solid};
    \end{semilogyaxis}
  \end{tikzpicture}
 \end{subfigure}
}

\newcommand{\plotprior}[2]{
  \begin{subfigure}{0.3\linewidth}
    \centering
    \begin{tikzpicture}[scale=0.6]
    \tikzstyle{every node}=[font=\Large]
    \pgfsetplotmarksize{1.2pt}
    \begin{axis}[%
        xmin=1, xmax=50,
        enlargelimits=0,
        axis on top,
        xlabel={Trial -- \textsc{#1} k=#2},
        ylabel={Influence Spread},
        legend columns = 4,
        legend entries = {\texttt{MLE}, \texttt{Beta(1, 999)}, \texttt{Beta(1, 99)}, \texttt{Beta(1, 9)}},
        legend to name = legprio,
        cycle list name = {color}
      ]
      \plreal{#1}{mle}{#2}{1.19.0.0.0}{re}{*}{solid};
      \plreal{#1}{oim}{#2}{1.999.0.0.0}{re}{|}{dashed};
      \plreal{#1}{oim}{#2}{1.99.0.0.0}{re}{star}{dashed};
      \plreal{#1}{oim}{#2}{1.9.0.0.0}{re}{triangle}{dashed};
    \end{axis}
  \end{tikzpicture}
 \end{subfigure}
}

\newcommand{\ploteps}[2]{
  \begin{subfigure}{0.3\linewidth}
    \centering
    \begin{tikzpicture}[scale=0.6]
    \tikzstyle{every node}=[font=\Large]
    \pgfsetplotmarksize{1.2pt}
    \begin{axis}[%
        xmin=1, xmax=50,
        enlargelimits=0,
        axis on top,
        xlabel={trial -- \textsc{#1} k=#2 $\epsilon$-greedy},
        ylabel={real spread (nodes)},
        legend columns = 3,
        legend entries = {\texttt{MLE}, \texttt{MLE+CB, $\epsilon=0.1$},
                          \texttt{MLE+EG}},
        legend to name = legeps,
        cycle list name = {color}
      ]
      \plreal{#1}{mle}{#2}{1.19.0.0.0}{re}{*}{dashed};
      \plreal{#1}{mle}{#2}{1.19.0.0.0.1}{re}{square}{dashed};
      \plreal{#1}{eg}{#2}{1.19.0}{re}{o}{solid};
    \end{axis}
  \end{tikzpicture}
 \end{subfigure}
}

\newcommand{\plotmse}[1]{
  \begin{subfigure}{0.4\linewidth}
    \centering
    \begin{tikzpicture}[scale=0.7]
    \tikzstyle{every node}=[font=\Large]
    \pgfsetplotmarksize{1.2pt}
    \begin{axis}[%
        ymin=0.06, ymax = 0.09,
        enlargelimits=0,
        axis on top,
        xlabel={trial -- \textsc{#1}},
        ylabel={mean sq. error},
        grid=both,
        legend columns = 3,
        legend entries = {$k=1$ prior, $k=1$ {\tt local}, $k=1$ {\tt MLE},
                          $k=5$ prior, $k=5$ {\tt local}, $k=5$ {\tt MLE},
                          $k=10$ prior, $k=10$ {\tt local}, $k=10$ {\tt MLE}},
        legend to name = legmse,
        cycle list name = {color}
      ]
      \plrealx{#1}{mse}{1}{0.0}{err}{o}{densely dotted};
      \plrealx{#1}{mse}{1}{1.0}{err}{o}{densely dashed};
      \plrealx{#1}{mse}{1}{1.2}{err}{o}{solid};
      \plrealx{#1}{mse}{5}{0.0}{err}{|}{densely dotted};
      \plrealx{#1}{mse}{5}{1.0}{err}{|}{densely dashed};
      \plrealx{#1}{mse}{5}{1.2}{err}{|}{solid};
      \plrealx{#1}{mse}{10}{0.0}{err}{star}{densely dotted};
      \plrealx{#1}{mse}{10}{1.0}{err}{star}{densely dashed};
      \plrealx{#1}{mse}{10}{1.2}{err}{star}{solid};
    \end{axis}
  \end{tikzpicture}
 \end{subfigure}
}


\section{Experimental Evaluation}\label{sec:experiments}

\begin{figure*}[t!]
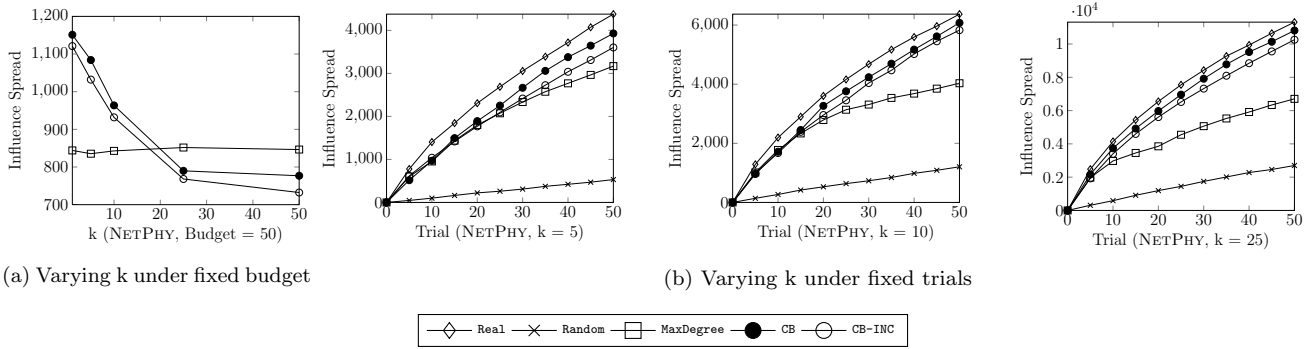

\mbox{
  \begin{subfigure}{0.23\linewidth}
    \centering
    \plotopt{NetPhy}{1}{50}
    \caption{Varying k under fixed budget} \label{fig:fixb}
 \end{subfigure}

  \begin{subfigure}{0.75\linewidth}
    \centering
    \plotKT{NetPhy}{5}{50}
    \quad
    \plotKT{NetPhy}{10}{50}
    \quad
    \plotKT{NetPhy}{25}{50}
    \caption{Varying k under fixed trials} \label{fig:fixk}
 \end{subfigure}
}
    \begin{center}
    {\tiny \ref{legkt}}
    \end{center}
    \vspace{-12pt}
   \caption{Heuristic-based v.s. Explore--Exploit.}
\end{figure*}

We now present the results. Section~\ref{subsec:setup} describes the experiment settings. In Sections~\ref{subsec:netphy} and \ref{subsec:allresult} we discuss our results on different datasets.


\subsection{Setup}
\label{subsec:setup}

We developed a ``real-world simulator'' to mimic the user feedback process of Figure~\ref{sf}.  This simulator first uses a
real social network to obtain a graph $G$.  It then associates an influence probability to each edge in $G$, where $p_{ij} = 1 / d_j$,
with $d_j$ the in-degree of node $j$. This setting of influence probability values is adopted in previous works~\cite{Kempe:2003, wei09, wei10, Goyal:2011, celf, TIM}.


When the chosen seed nodes are tested on whether they can influence other nodes, the simulator runs a single independent cascade simulation on $G$, and obtains feedback information $F_n$, in a form of $(i,j,a_{ij})$ and $A_n$, the set of successfully activated nodes. We measure the effectiveness of an OIM solution by its influence
spread in the real world, after $N$ trials, as the total number of
successfully activated nodes in these trials, i.e, $|\cup_{n=1}^{N}{A_n}|$. 
We repeat each solution $10$ times and report the average.

\begin{table}
  \centering
  \caption{Datasets\label{tab:datasets}}
  \begin{tabular}{lccc}
    \toprule
    Dataset & {\sc NetHEPT} & {\sc NetPHY} & {\sc DBLP}\\
    \midrule
    \texttt{\#} of Nodes & 15K & 37K & 655K \\
    \texttt{\#} of Edges & 59K & 231K & 2.1M \\
    avg{.} degree & 7.73 & 12.46 & 6.1 \\
    max{.} degree & 341 & 286 & 588 \\
    \bottomrule
  \end{tabular}
\end{table}

{\bf Datasets.} We have studied several real social network datasets. We have used {\sc NetHept} and {\sc NetPhy} are
collaboration networks, obtained from arXiv.org in the High Energy Physics Theory and Physics domains, respectively. We have also used the {\sc DBLP} graph, which is an academic collaboration network.
In these datasets, nodes represent authors, and edges
representing co-authorship.  These datasets are commonly used in the literature of influence maximization~\cite{Kempe:2003, wei09, wei10, Goyal:2011, TIM}. Table~\ref{tab:datasets} shows the details of these datasets.

{\bf Options for OIM algorithm.} We have evaluated several possible options for the {\it seed selection} and {\it graph update} components for our OIM solution:

{\bf[Choosing seeds]}
        \begin{compactitem}
            \item Heuristic-based strategies: \texttt{Random}, \texttt{MaxDegree};
            \item Explore--Exploit strategies:
                1) \texttt{Exploit} contains only exploit algorithm;
                2) \texttt{$\epsilon$-greedy} represents $\epsilon$-greedy algorithm;
                3) \texttt{CB} is our Confidence-Bound explore--exploit algorithm with Exponentiated Gradient update.
        \end{compactitem}

{\bf[Updating graph]}
        \begin{compactitem}
            \item \texttt{NO} does not conduct any update;
            \item \texttt{LOC} only local updates;
            \item \texttt{LSE} local and global updates where Least Squares Estimation is adopted in global update;
            \item \texttt{MLE} as {\tt LSE}, but Maximum Likelihood Estimation is adopted.
        \end{compactitem}

In our experiments, we compare the algorithms using combinations of the above
two components. Note that {\tt Random} and {\tt MaxDegree} do not rely on the
influence probability of the edges, and they are not combined with update
methods.  When a particular EE strategy is adopted, the update method would be
specified, for instance, \texttt{CB+MLE} means that we use \texttt{CB} with
\texttt{MLE} update.  By default, we use \texttt{MLE} for updating the graph.
Furthermore, if the EE strategy is used in choosing seeds, we use \texttt{CB} by
default.

When an IM algorithm is invoked in an EE strategy, we use TIM+ since it is the
state-of-art influence maximization algorithm.  We also compare the incremental
approach with the non-incremental one for EE strategy. For example, we denote
the incremental version for \texttt{CB} as \texttt{CB-INC}.

{\bf Parameters.}
By default, the global prior is set to be $\mathrm{B}(1,19)$,
$\boldsymbol{\theta}=\{-1,0,1\}$ in \texttt{CB}, $\epsilon = 0.1$ in
\texttt{$\epsilon$-greedy}, and $\tau=0.02$ in the incremental approach.

Our algorithms, implemented in C++, are conducted on a Linux machine with a 3.40\,GHz Octo-Core Intel(R) processor and 16GB of memory. Next, we focus on {\sc NetPHY}, and evaluate different
combinations of the algorithms in our OIM framework. We summarize our results for other datasets in Section~\ref{subsec:allresult}.

\subsection{Results on NetPHY}
\label{subsec:netphy}

{\bf Heuristic-based v.s. Explore--Exploit.}
We first fix the total budget and verify how the OIM algorithms perform with
different number of trials. We set $Budget=50$, and vary $k$ in $\{1, 5, 10, 25,
50\}$. By varying $k$, we essentially vary the total budget. For example, with
$k=5$, $50$ units of budget is invested over $N=10$ trials.
Figure~\ref{fig:fixb} shows our results. Since \texttt{Random} only has
influence spread less than $200$ on average, we do not plot it. We observe that
the spread of \texttt{MaxDegree} does not change much since it does not depend
on the real-world feedback. For \texttt{CB}, its spread increases when $k$
decreases and it is better than \texttt{MaxDegree} when $k \leq 10$ (or $N \ge
5$). Specifically, when $k=1$, \texttt{CB} is about $35\%$ better than
\texttt{MaxDegree}.  The reason is that, for \texttt{CB}, a smaller $k$
indicates more chances to get real-world feedback, and thus, more chances to
learn the real influence graph, which leads to a better result.  Moreover, when
$k=50$, all budget is invested once, which can be regarded as an offline
solution, and produces the worst result for \texttt{CB}. This further indicates
the effectiveness of our OIM framework.  For \texttt{CB-INC}, it performs close
to \texttt{CB} with only a small discount (around $5\%$ for different $k$) on
the spread. It supports our claim that the incremental approach can perform
without incurring much error.

We next fix $k$ and compare different algorithms in Figure~\ref{fig:fixk}.  The
results are consistent with our previous findings that \texttt{CB} outperforms
other variants. \texttt{CB-INC} produces similar results with \texttt{CB}. We
observe that the gap between \texttt{CB} and \texttt{MaxDegree} increases with
$N$ and $k$. For example, at $N=50$, \texttt{CB} is about $20\%$ better than
\texttt{MaxDegree} when $k=5$, and the percentage grows to $45\%$ when $k=25$.
The reason is that larger $k$ and larger $N$ give more chances for \texttt{CB}
to learn the real influence graph.  We also plot the result for TIM+ when the
real influence probability is known, denoted as {\tt Real}. This can be seen as
an oracle, serving as a reference for other algorithms. We find that {\tt CB}
performs close to {\tt Real}, and its discount on the spread decreases with $N$.
For example, when $k=5$, the discount decreases from $30\%$ at $N=10$ to $13\%$
at $N=50$. This indicates that, with more real-world feedback, the learned graph
for {\tt CB} is closer to the real graph, and thus, leads to a closer result to
{\tt Real}.

\begin{figure} [t!]
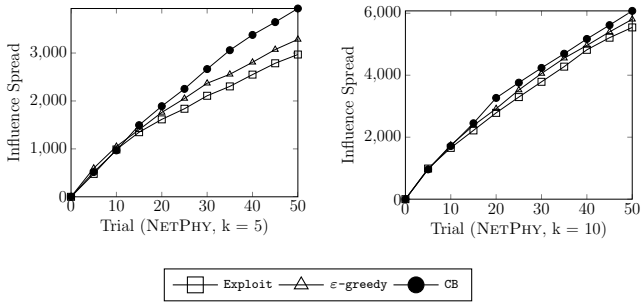

     \centering

    \plotee{NetPhy}{5}{50}
    \quad
    \plotee{NetPhy}{10}{50}

    \begin{center}
    {\tiny \ref{legee}}
    \end{center}
    \vspace{-12pt}
    \caption{Explore--exploit strategies\label{fig:eps}}
\end{figure}

\textbf{ Explore--Exploit Strategies.}
We compare three versions of the EE strategies for different $k$ in
Figure~\ref{fig:eps}.  We observe that {\tt Exploit} is the worst, since it may
suffer from the wrong prediction of the influence probabilities and does not
explore other potential high influencing nodes. {\tt CB} is the best,
especially, for small $k$.  When $k=5, N=50$, {\tt CB} is about $20\%$ and
$32\%$ better than {\tt $\epsilon$-greedy} and {\tt Exploit}, respectively. The
reason is that for a smaller $k$, fewer feedback tuples are returned in one
trial, which makes the learned influence graph converge to the real graph
slower. Hence, the effect of exploration is strengthened, which is more
favorable to {\tt CB}. We have also conducted experiments for $\epsilon$-greedy
by varying $\epsilon$. We observe that its performance is sensitive to
$\epsilon$ and $\epsilon=0.1$ is the best one in our results, but it is still
worse than {\tt CB} in all cases.

\begin{figure}
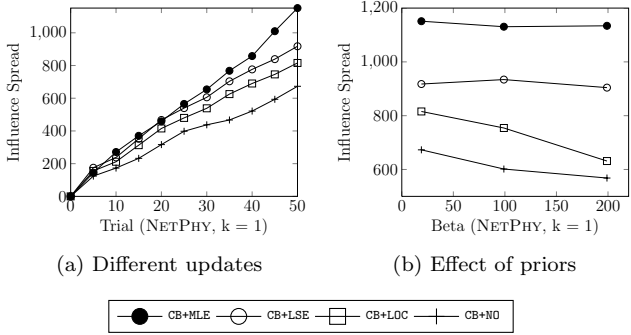

\mbox{
  \begin{subfigure}[b]{0.5\linewidth}
    \centering
    \plotupdates{NetPhy}{1}{50}

    \caption{Different updates} \label{fig:up}
 \end{subfigure}

  \begin{subfigure}[b]{0.5\linewidth}
    \centering
     \plotpriors{NetPhy}{1}

    \caption{Effect of priors} \label{fig:priors}
 \end{subfigure}
}
    \begin{center}
    {\tiny \ref{legupdates}}
    \end{center}
\vspace{-12pt}
   \caption{Comparing different updating methods}\label{fig:ups}
\end{figure}

\textbf{Updating the uncertain influence graph.}
In Figure~\ref{fig:up}, we compare different updating methods for the uncertain
influence graph.  Although {\tt NO} makes use of the prior knowledge about the
influence graph to select seeds, it still performs worse than other update
options.  {\tt LOC} is slightly better, but still worse than {\tt MLE} and {\tt
LSE}, since it does not employ any global update and it suffers from the
sparseness of the activations.  {\tt MLE} is the best (about $25\%$ better than
{\tt LSE} and $40\%$ better than {\tt LOC}), which is consistent with the fact
that {\tt MLE} makes use of the full feedback to update the graph while {\tt
LSE} only utilizes the set of successfully activated nodes.

We also test the updating methods with different priors
(Figure~\ref{fig:priors}) to check whether they are sensitive to the prior. We
observe that while {\tt LOC} and {\tt NO} fluctuate a lot with different priors,
{\tt MLE} and {\tt LSE}'s performance is very stable. In fact, during different
runs of {\tt MLE} and {\tt LSE} with different priors, the global $\beta$
values all converge to around $27$. This supports the fact that the global
updating techniques are crucial when we do not have good prior information. Even
an inexact choice of prior will be generally fixed, minimizing the impact on
performance.

\textbf{Efficiency.}
In Figure~\ref{fig:running}, we illustrate the cumulative running time for
running $N$ trials for different algorithms. {\tt Random} and {\tt MaxDegree}
are most efficient as they do not rely on any influence evaluation. With the
help of incremental approach, {\tt CB-INC} runs significantly faster than {\tt
CB}, and for the case where $N\geq 10$, it achieves about 10 times speedup. For
instance, at $N=50$, {\tt CB-INC} reduces the running time by $88\%$, compared
to {\tt CB}. This is intuitive, as in the first few trials the graph is more
uncertain, and the updates affect the samples a lot. However, when $N\geq10$, we
observe that the global priors become more stable, leading to a high ratio of
re-using samples (e.g., the ratio is about $80\%$ to $99\%$ when $N\geq 10$).
Moreover, the average in-degree of {\sc NetPHY} is 12.46, making the time of
generating a new sample about an order of magnitude slower than re-using a
sample. These two factors together make {\tt CB-INC} have a much more efficient
performance than {\tt CB}.

We then show the efficiency results by fixing $Budget=50$ and varying $k$ in
Figure~\ref{fig:runningK}. The running time of {\tt MaxDegree} and {\tt Random}
is stable for various $k$, while {\tt CB} and {\tt CB-INC} show a decline on
efficiency when $k$ decreases. This is because a smaller $k$ indicates that more
trials are required to invest all budget, and so, TIM+ should be executed more
often, for a general decrease in efficiency. Another observation is that the
improvement of {\tt CB-INC} over {\tt CB} increases with $k$. This further
strengthens the utility of using {\tt CB-INC} in practice.
Figure~\ref{fig:runningK} and Figure~\ref{fig:fixb} together show a tradeoff of
setting $k$: a smaller $k$ leads to a better performance in spread but worse
performance in efficiency. We suggest to set a small $k$ to ensure the
algorithm's better performance in spread. The value of $k$ will depend on how
much total time that the user can afford.

\begin{figure}[t!]
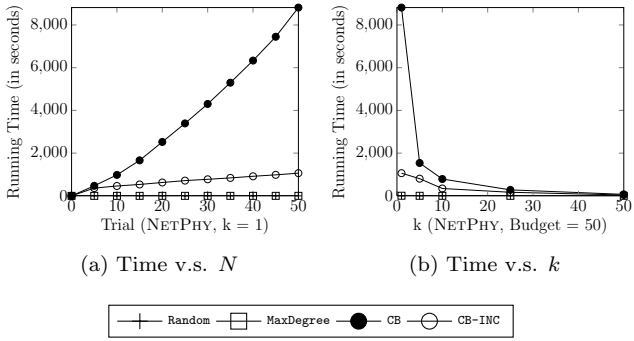


\mbox{
  \begin{subfigure}{0.5\linewidth}
    \centering
    \plott{NetPhy}{1}{50}
    \caption{Time v.s. $N$}
    \label{fig:running}
 \end{subfigure}

  \begin{subfigure}{0.5\linewidth}
    \centering
    \plottt{NetPhy}
    \caption{Time v.s. $k$}
    \label{fig:runningK}
  \end{subfigure}
}

    \begin{center}
    {\tiny \ref{legtt}}
    \end{center}
    \vspace{-12pt}
    \caption{Cumulative running time}
\end{figure}

\textbf{Effect of $\tau$.}
We also verify the effect of $\tau$ in the incremental approach by varying
$\tau$ from $0.01$ to $0.03$ and fixing $k=1, Budget=50$. We compare them with
{\tt CB}, the non-incremental algorithm. First, a smaller $\tau$ gives better
results in terms of influence spread. For instance, it leads to 3\%, 5\%, 15\%
discount in spread compared with {\tt CB} for $\tau = 0.01, 0.02, 0.03$,
respectively. However, a smaller $\tau$ leads to a slowdown in efficiency since
it has a stricter requirement in global check. For example, the running time for
$\tau=0.01$ is about $28\%$ slower than the one for $\tau=0.02$ and $38\%$ worse
than the one for $\tau=0.03$.

\textbf{Discussion.}
The OIM framework is highly effective in maximizing influence when the real
influence probabilities are unknown.  In this framework, {\tt MLE} is the best
updating method. Moreover, {\tt CB} and {\tt CB-INC} consistently outperform
other algorithms. By using {\tt CB-INC}, we can also significantly improve the
efficiency of {\tt CB}, with only a small discount in influence spread.

\subsection{Results for NetHEPT and DBLP}
\label{subsec:allresult}

\begin{figure}[t]
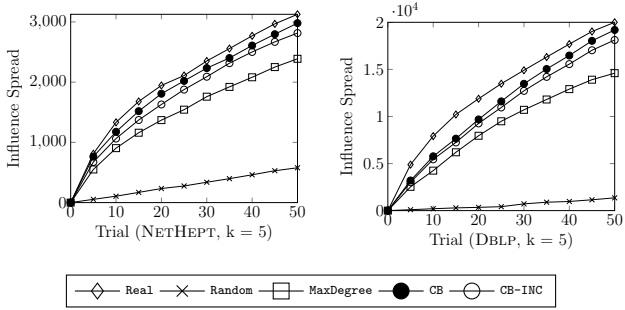

\mbox{
  \begin{subfigure}{0.5\linewidth}
    \centering
    \plotKTsmall{NetHept}{5}{50}
 \end{subfigure}

  \begin{subfigure}{0.5\linewidth}
    \centering
    \plotKTsmall{Dblp}{5}{50}
  \end{subfigure}
}
    \begin{center}
    {\tiny \ref{legkt}}
    \end{center}
   \caption{Effectiveness on other datasets}
   \label{fig:effectall}
\end{figure}

Figure~\ref{fig:effectall} and Figure~\ref{fig:effiall} show representative
results for {\sc NetHEPT} and {\sc DBLP}. These results are consistent with the
ones for {\sc NetPHY}, where {\tt CB} and {\tt CB-INC} are close to the oracle
({\tt Real}), and better than heuristic-based algorithms in maximizing influence
spread. For efficiency, {\tt CB-INC} significantly reduces the running time of
{\tt CB}, especially for a large dataset {\sc DBLP}. For instance, at
$k=1,N=50$, {\tt CB-INC} saves 16 hours compared with {\tt CB} which costs 19
hours in total to get the result for {\sc DBLP}.


\section{Conclusions}\label{sec:conc}

In this paper, we examine how to perform influence maximization when influence
probabilities may not be known in advance. We develop a new solution, where IM
is performed in multiple trials, and we have proposed explore--exploit strategies for this problem. We showed
experimentally that explore--exploit based on the uncertainty in the graph performs
well. We also proposed novel methods to update
the knowledge of the graph based on the feedback received from the real world, and showed experimentally that they are effective in longer campaigns.  Even when the influence probabilities are not known in advance, the influence spread of our solution is close to the spread using the real influence graph, especially when the number of trials increases.

In the future, we will examine the scenario where budgets are different in each
trial. We will extend our solution to handle other complex situations (e.g., the
change of influence probability values over time), consider IM methods
(e.g.,~~\cite{cbtopkin}, \cite{INFLEX}) that utilize community and topic
information, and other influence propagation models, such as linear threshold or
credit distribution~\cite{Goyal:2011, Singer:2012, siandposteriore}. Another
direction is to increase the scalability of our methods; this may require
distributed algorithm, such as distributed sampling.

\begin{figure}[t]
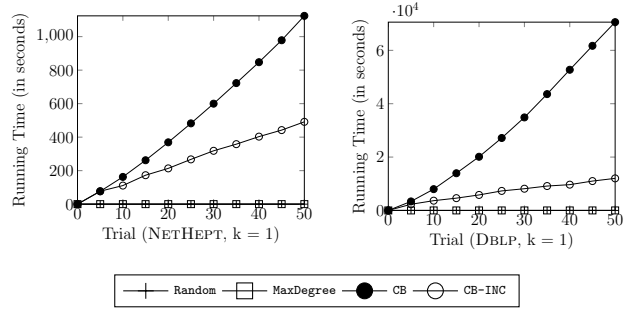

\mbox{
  \begin{subfigure}{0.5\linewidth}
    \centering
    \plottsmall{NetHept}{1}{50}
 \end{subfigure}

  \begin{subfigure}{0.5\linewidth}
    \centering
    \plottsmall{Dblp}{1}{50}
  \end{subfigure}
}
    \begin{center}
    {\tiny \ref{legtt}}
    \end{center}
   \caption{Efficiency on other datasets}
   \label{fig:effiall}
\end{figure}

\bibliography{oim}  

\begin{thebibliography}{10}

\bibitem{Agrawal:2012}
S.~Agrawal and N.~Goyal.
\newblock Analysis of thompson sampling for the multi-armed bandit problem.
\newblock In {\em COLT}, 2012.

\bibitem{INFLEX}
C.~Aslay, N.~Barbieri, F.~Bonchi, and R.~A. Baeza-Yates.
\newblock Online topic-aware influence maximization queries.
\newblock In {\em EDBT '14}, 2014.

\bibitem{UCB}
P.~Auer.
\newblock Using confidence bounds for exploitation-exploration trade-offs.
\newblock {\em JMLR}, 3:397--422, 2003.

\bibitem{Borgs}
C.~Borgs, M.~Bratbar, J.~Chayes, and B.~Lucier.
\newblock Maximizing social influence in nearly optimal time.
\newblock In {\em SODA 2014}.

\bibitem{Cesa}
N.~Cesa-Bianchi and G.~Lugosi.
\newblock {\em Prediction, Learning, and Games}.
\newblock Cambridge University Press, 2006.

\bibitem{wei10}
W.~Chen, C.~Wang, and Y.~Wang.
\newblock Scalable influence maximization for prevalent viral marketing in
  large-scale social networks.
\newblock In {\em KDD}, 2010.

\bibitem{wei09}
W.~Chen, Y.~Wang, and S.~Yang.
\newblock Efficient influence maximization in social networks.
\newblock In {\em KDD}, 2009.

\bibitem{CMAB}
W.~Chen, Y.~Wang, and Y.~Yuan.
\newblock Combinatorial multi-armed bandit: General framework and applications.
\newblock In {\em ICML}, 2013.

\bibitem{wei10lt}
W.~Chen, Y.~Yuan, and L.~Zhang.
\newblock Scalable influence maximization in social networks under the linear
  threshold model.
\newblock In {\em ICDM}, 2010.

\bibitem{domingos01}
P.~Domingos and M.~Richardson.
\newblock Mining the network value of customers.
\newblock In {\em KDD}, 2001.

\bibitem{Goyal:2011}
A.~Goyal, F.~Bonchi, and L.~V. Lakshmanan.
\newblock {A data-based approach to social influence maximization}.
\newblock {\em VLDB 2011}.

\bibitem{Goyal:2010}
A.~Goyal, F.~Bonchi, and L.~V. Lakshmanan.
\newblock {Learning influence probabilities in social networks}.
\newblock In {\em WSDM 2010}.

\bibitem{celfplus}
A.~Goyal, W.~Lu, and L.~V. Lakshmanan.
\newblock Celf++: Optimizing the greedy algorithm for influence maximization in
  social networks.
\newblock In {\em WWW 2011}.

\bibitem{siandposteriore}
J.~Huang, X.-Q. Cheng, H.-W. Shen, T.~Zhou, and X.~Jin.
\newblock Exploring social influence via posterior effect of word-of-mouth
  recommendations.
\newblock In {\em WSDM}, 2012.

\bibitem{IRIE}
K.~Jung, W.~Heo, and W.~Chen.
\newblock Irie: Scalable and robust influence maximization in social networks.
\newblock In {\em ICDM '12}, 2012.

\bibitem{Kempe:2003}
D.~Kempe, J.~Kleinberg, and E.~Tardos.
\newblock {Maximizing the spread of influence through a social network}.
\newblock In {\em KDD 2003}.

\bibitem{IPA}
J.~Kim, S.-K. Kim, and H.~Yu.
\newblock Scalable and parallelizable processing of influence maximization for
  large-scale social networks?
\newblock In {\em ICDE 2013}.

\bibitem{celf}
J.~Leskovec, A.~Krause, C.~Guestrin, C.~Faloutsos, J.~VanBriesen, and
  N.~Glance.
\newblock Cost-effective outbreak detection in networks.
\newblock In {\em KDD 2007}.

\bibitem{marketing}
M.~J. Lovett, R.~Peres, and R.~Shachar.
\newblock On brands and word of mouth.
\newblock {\em J. Marketing Research}, 50(4), 2013.

\bibitem{Lu:2013}
W.~Lu, F.~Bonchi, A.~Goyal, and L.~V. Lakshmanan.
\newblock The bang for the buck: Fair competitive viral marketing from the host
  perspective.
\newblock KDD '13, New York, NY, USA. ACM.

\bibitem{rl}
S.~Richard and A.~Barto.
\newblock {\em Reinforcement Learning: An Introduction}.
\newblock MIT Press, 1998.

\bibitem{domingos02}
M.~Richardson and P.~Domingos.
\newblock Mining knowledge-sharing sites for viral marketing.
\newblock In {\em KDD 2002}.

\bibitem{Robbins:1952}
H.~Robbins.
\newblock {Some aspects of the sequential design of experiments}.
\newblock {\em Bull. Amer. Math. Soc.}, 58(5), 1952.

\bibitem{saito}
K.~Saito, R.~Nakano, and M.~Kimura.
\newblock Prediction of information diffusion probabilities for independent
  cascade model.
\newblock In {\em KES}, 2008.

\bibitem{Singer:2012}
Y.~Singer.
\newblock {How to win friends and influence people, truthfully: influence
  maximization mechanisms for social networks}.
\newblock In {\em WSDM}, 2012.

\bibitem{TIM}
Y.~Tang, X.~Xiao, and Y.~Shi.
\newblock Influence maximization: Near-optimal time complexity meets practical
  efficiency.
\newblock In {\em SIGMOD 2014}.

\bibitem{cbtopkin}
Y.~Wang, G.~Cong, G.~Song, and K.~Xie.
\newblock Community-based greedy algorithm for mining top-k influential nodes
  in mobile social networks.
\newblock In {\em KDD}, 2010.

\end{thebibliography}
\bibliographystyle{abbrv}

\end{document}